\documentclass[11pt, oneside]{amsart}
\usepackage[marginparwidth=2cm]{geometry}                
\usepackage{graphicx, slashed}
\usepackage{amssymb, amsthm, mathrsfs, bm}
\usepackage{epstopdf}
\usepackage{enumerate}
\usepackage{color}
\usepackage{hyperref}
\newcommand{\be}{\begin{equation*}}
\newcommand{\ee}{\end{equation*}}
\newcommand{\ben}[1]{\begin{equation}\label{#1}}
\newcommand{\een}{\end{equation}}
\newcommand{\bea}{\begin{eqnarray}}
\newcommand{\eea}{\end{eqnarray}}
\newcommand{\bean}{\begin{eqnarray*}}
\newcommand{\eean}{\end{eqnarray*}}

\newcommand{\R}{\mathbb{R}}

\newcommand{\p}{\partial}

\renewcommand{\O}[1]{\mathcal{O}\left( #1 \right)}

\newcommand{\abs}[1]{\left|#1 \right|} 
\newcommand{\norm}[2]{\left|\left |#1 \right| \right |_{#2}}

\newcommand{\eq}[1]{({\ref{#1}})}
\newtheorem{Theorem}{Theorem}

\newtheorem{Lemma}{Lemma}

\numberwithin{Theorem}{section}
\numberwithin{Lemma}{section}
\numberwithin{Proposition}{section}

\newcommand{\gr}{\nabla}
\newcommand{\grt}{\tilde{\nabla}}
\newcommand{\temt}{\tilde{\mathbb{T}}}

\newcommand{\pa}{\partial}

\newcommand{\hf}{\frac{1}{2}}
\def\d[#1]{\text{d} #1}

\def\pd[#1]{\frac{\partial}{\partial {#1}}}
\def\pdt[#1][#2]{\frac{\partial {#1}}{\partial {#2}}}
\newcommand{\td}{\dot{t}}
\newcommand{\rd}{\dot{r}}
\newcommand{\xd}{\dot{x}}
\newcommand{\yd}{\dot{y}}

\newcommand{\mf}{\mathcal{M}}
\newcommand{\mh}{\mathcal{H}}

\newcommand{\mi}{\mathcal{I}}
\newcommand{\mo}{\mathcal{O}}
\newcommand{\su}{\slashed{\gr}u}
\newcommand{\Lu}{\underline{L}}
\newcommand{\Hu}{\underline{H}}
\newcommand{\cu}{\chi u}
\newcommand{\gd}{\dot{\gamma}}
\makeatletter
\newtheorem*{rep@theorem}{\rep@title} \newcommand{\newreptheorem}[2]{%
\newenvironment{rep#1}[1]{%
\def\rep@title{\bf #2 \ref{##1} }%
\begin{rep@theorem} }%
{\end{rep@theorem} } }
\makeatother
\newreptheorem{theorem}{Theorem}
\newreptheorem{lemma}{Lemma}

\newreptheorem{prop}{Proposition}

\newtheorem{rem}{Remark}
\numberwithin{rem}{section}

\newcommand{\f}{\frac}

\title[]{The Klein-Gordon equation on the Toric AdS-Schwarzschild Black Hole}
\author{Jake Dunn}

\author{Claude Warnick}
\thanks{\vspace{.1cm}\texttt{j.dunn15@imperial.ac.uk};  \texttt{c.warnick@imperial.ac.uk}\\
\phantom{1    }\hspace{.05cm} Dept. of Mathematics, South Kensington Campus, Imperial College London, SW7 2AZ, UK.\vspace{.1cm}\\
\phantom{1    }\hspace{.05cm} Mathematics Institute,  Zeeman Building, University of Warwick, Coventry, CV4 7AL, UK\vspace{.1cm}}

\date{\today \vspace{.1cm}}  

\begin{document}
\maketitle
\begin{abstract}
We consider the Klein-Gordon equation on the exterior of the toric anti de-Sitter Schwarzschild black hole with Dirichlet, Neumann and Robin boundary conditions at $\mi$. We define a non-degenerate energy for the equation which controls the renormalised $H^1$ norm of the field. We then establish both decay and integrated decay of this energy through vector field methods. Finally we demonstrate the necessity of `losing a derivative' in the integrated energy estimate through the construction of a Gaussian beam staying in the exterior of the event horizon for arbitrary long co-ordinate time.  
\end{abstract}
\section{Introduction}

Among asymptotically flat spacetimes satisfying the vacuum Einstein equations:
\begin{equation}\label{fEVE}
Ric_g = 0,
\end{equation}
the black hole solutions occupy a privileged position. It is conjectured that the family of rotating black holes described by the Kerr metric represent the final state of gravitational collapse. Within the class of spherically symmetric solutions with scalar matter it is possible to establish this fact rigorously \cite{csgsf}. Establishing such a result without an assumption of symmetry appears a very challenging task. Less ambitiously, one might hope to show that the Kerr family of black holes are in some suitable sense \emph{stable} as solutions of \eq{fEVE} against small perturbations to the initial data. In the absence of a black hole, the nonlinear stability of the Minkowski spacetime against small perturbations was established in the work of Christodoulou-Klainerman \cite{CKMS}. For black hole spacetimes, the nonlinear stability against generic small perturbations remains an important open problem, although see \cite{Holzegel:2010wm, Dafermos:2013bua} for an alternative approach.

In the linear setting, considerable progress towards establishing the stability of the black hole spacetimes has been made over the last few years, culminating in the recent proof of decay for solutions of the scalar wave equation on any fixed subextremal Kerr black hole background \cite{drsrK}. Even this highly simplified problem nevertheless requires a very careful analysis. The key issues to be understood are those of \emph{trapping} and \emph{superradiance}. The issue of trapping is related to the presence of null geodesics which orbit the central black hole, neither escaping to infinity, nor falling through the black hole horizon. Trapped null geodesics are an obstacle to linear decay, resulting in decay estimates that `lose derivatives'. The other problem to be confronted is the superradiant effect, which manifests itself in the absence of a globally timelike Killing field. As a consequence, the natural conserved `energy' is not coercive and does not control the solution. Overcoming these issues requires a rather subtle argument involving the nature of the set of trapped geodesics and its interplay with the superradiance phenomenon.

In this paper, we shall initiate the mathematical study of the Klein-Gordon equation on a family of black hole backgrounds which, largely, evade the complications of trapping and superradiance present for the Kerr family of black holes: the planar, or toroidal, AdS-Schwarzschild black holes. These are solutions of the vacuum Einstein equations with a negative cosmological constant:
\begin{equation}\label{EVE}
Ric_g = \Lambda g, \qquad \Lambda <0.
\end{equation}
The presence of a negative cosmological constant for our purposes has two main effects. Firstly, the character of null infinity changes: it becomes a timelike surface, meaning that it is necessary to impose boundary conditions in order for the Klein-Gordon equation to define a well posed evolution problem. Secondly, for  $\Lambda<0$ it is no longer the case that the horizon of a compact black hole in four dimensions must be spherical: versions of the AdS-Schwarzschild black holes exist whose horizon is a surface of arbitrary genus. We shall consider the case of a toroidal AdS-Schwarzschild black hole, with isometry group $\R \times T^2$. The spherical case has been studied in \cite{Holzegel:2013kna,Holzegel:2011uu}.

The main difference between the toric AdS-Schwarzshild black holes and the spherical AdS-Schwarzschild black hole can be seen at the level of the null geodesics. In the equation governing null geodesics in the toric black hole background, there is no analogue of the centrifugal force that appears for the spherical black hole. Crudely, there exist trapped null geodesics in the spherically symmetric case owing their existence to a combination of the centrifugal repulsion and the confining nature of the  AdS boundary. For the toric black hole, the centrifugal force is absent, and so no trapped null geodesics exist. This has a profound effect on the behaviour of solutions to the Klein-Gordon equation on this background.

Aside from the fact that solutions of \eq{EVE} are of interest in the context of classical relativity, they have also received considerable attention in the context of the putative AdS/CFT correspondence. This conjectures a relation between gravitational theories with negative cosmological constant, and conformal field theories in one fewer dimension. In this context linear fields on the planar\footnote{The difference between planar and toric AdS/Schwarzschild is purely topological: the toric solution arises by periodically identifying the planes of symmetry present in the planar solution.} AdS-Schwarzschild solution are studied, as they are conjectured to be dual to conformal field theories in Minkowski space at finite temperature.

\subsection*{The main results}

We shall first consider solutions of the Klein-Gordon equation:
\begin{equation}\label{introKGE}
\Box_g u + \frac{\alpha}{l^2}u = 0,
\end{equation}
where $\alpha < \frac{9}{4}$. (When dealing with integrated decay due to a technical limitation we will restrict to $\alpha^*<\alpha<\frac{9}{4}$ where $\alpha^*$ is a constant with the property $\alpha^*\in \left(\frac{5}{4},\frac{27}{{16}} \right)$, ($\alpha^* \approx 1.46$). Our first result is an extension of the results of \cite{warn2} to the toric AdS-Schwarzschild black hole. That is to say, for a stationary time foliation which is regular at the horizon, we shall establish the following result:
\begin{reptheorem}{thm1}
Suppose $u$ is a solution to \eqref{KGE} satisfying suitable (Dirichlet, Neumann or Robin) boundary conditions at infinity. Let $\mathcal{E}[u]$ be the renormalised energy density of the field $u$. Then for any $T>0$ we have:
\be
\int_{\{t=T\}} \mathcal{E}[u] \lesssim  \int_{\{t=0\}}\mathcal{E}[u],
\ee
with the implicit constant independent of $T$.
\end{reptheorem}
Here and elsewhere, $\mathcal{E}[u]$ is an energy density which does \emph{not} degenerate at the event horizon, so we gain full control of all derivatives up to, and including, the horizon. Having established boundedness, we then consider the question of decay for solutions of \eq{introKGE}. We establish an integrated decay estimate


\begin{reptheorem}{thm2}
Suppose $u$ is a smooth solution to \eqref{KGE} satisfying Dirichlet, Neumann or Robin boundary conditions at infinity. Let $\mathcal{E}[u]$ be the renormalised energy density of the field $u$. Then for any $T>0$ we have:
\be
\int_{\{0<t< T\}} \frac{\mathcal{E}[u]}{r^3} \lesssim  \int_{\{t=0\}}\mathcal{E}[u],
\ee
with the implicit constant independent of $T$.
\end{reptheorem}
Here $r$ is a radial coordinate with null infinity corresponding to $r\to\infty$. Clearly, the weight on the left hand side of this estimate is weaker than that on the right. In fact, this degeneration occurs only for derivatives tangent to $\mi$. We can improve the weight at the cost of losing a derivative:
\begin{reptheorem} {thm3}
Suppose $u$ is a smooth solution to \eqref{KGE} satisfying the same boundary conditions as above. Let $\mathcal{E}[u]$ be the renormalised energy density of the field $u$. Then for any $T>0$ we have:
\be
\int_{\{0<t< T\}}{\mathcal{E}[u]} \lesssim  \int_{\{t=0\}}\left( \mathcal{E}[u] + \mathcal{E}[u_t]\right) ,
\ee
with the implicit constant independent of $T$.\\
\end{reptheorem}  
It is likely that one does not need to lose a whole derivative to improve the weight, but we shall not pursue this point here. 

In order to establish the results above, we combine the vector field method with the renormalisation methods of \cite{warn1}. The boundedness proof follows a very similar approach to that of \cite{warn2}. To establish decay, we make use of currents constructed from Morawetz vector fields. In contrast to the spherical black holes, the absence of a photon sphere simplifies the construction of appropriate currents.

Combining the integrated decay estimate of Theorem \ref{thm3} with the result of Theorem \ref{thm1}, and an argument based on the redshift, we can establish:
\begin{reptheorem}{thm3.5}
Suppose $u$ is a solution to \eqref{KGE} satisfying suitable (Dirichlet, Neumann or Robin) boundary conditions at infinity. Let $\mathcal{E}[u]$ be the renormalised energy density of the field $u$. Then for any $T>0$ we have:
\be
\int_{\{t=T\}} \mathcal{E}[u] \lesssim  \frac{1}{(1+T)^n} \sum_{k=0}^n\int_{\{t=0\}}\mathcal{E}\left [\p_t^ku \right],
\ee
with the implicit constant independent of $T$.
\end{reptheorem}

It is worth contrasting this result with that of \cite{Holzegel:2013kna,Holzegel:2011uu} for the spherical AdS-Schwarzschild and Kerr AdS black holes. They established that in this background solutions of the Klein-Gordon equation satisfy a decay estimate of the form:
\be
\int_{\{t=T\}} \mathcal{E}[u] \lesssim \frac{1}{\log(2+T)}  \int_{\{t=0\}}\left( \mathcal{E}[u] + \mathcal{E}[u_t]\right) .
\ee
and moreover this cannot be improved. The slower rate of decay here is due to the existence of null geodesics which \emph{never} cross the horizon: they are stably trapped between the centrifugal barrier and null infinity (which for these purposes may be thought of as a reflecting barrier).

The fact that our estimates exhibit a loss, either in weight or in derivatives, is due to a `trapping at infinity'. There exist null geodesics which start far from the black hole, moving orthogonally to the radial direction, and which take an arbitrarily long time to cross the horizon. Using this fact, one can establish using the Gaussian beam methods of \cite{GBS} that:
\begin{reptheorem}{thm4} 
There exists no constant $C>0$, independent of $T$, such that the estimate 
\be
\int_{\{0<t< T\}}{\mathcal{E}[u]} \leq C  \int_{\{t=0\}}\mathcal{E}[u],
\ee
holds for all smooth solutions $u$ of \eqref{KGE}.
\end{reptheorem}   
\noindent This result is very robust, in particular it does not depend on the boundary conditions, nor on the value of the Klein-Gordon mass.

\subsection*{Outline of the paper} This paper consists of four main sections. The first one defines the spacetime, sets up the appropriate hypersurfaces we will need and states the relevant divergence theorem. The second section sets up the initial boundary value problems (IBVP) we are interested in, the appropriate re-normalisation scheme for well posedness as seen in \cite{warn1} and defines the necessary tensorial quantities for proving decay. The third section proves energy decay for the IBVPs, taking the approach of \cite{warn2}. The final section then seeks to obtain qualitative time decay rates for the energy of solutions of the IBVP in the time coordinate using vector field methods found in \cite{mora1} and \cite{DRS}.

\section{Toric AdS Schwarzschild black hole}  
\subsection{The Manifold}
We define the exterior of the toric AdS-Schwarzchild black hole with mass $M$ and AdS radius $l$ to be the following manifold with boundary
\[
\mathcal{M} = \mathbb{R}_{t \ge 0} \times \mathbb{R}_{r\ge r_+}\times \mathbb{T}^2,
\]
with Lorentzian metric
\[
g = -\left( \frac{-2M}{r}+\frac{r^2}{l^2}\right)dt^2+\frac{4Ml^2}{r^3}dtdr +\left( \frac{2Ml^4}{r^5}+\frac{l^2}{r^2}\right) dr^2 + r^2(dx^2+dy^2). 
\]
Where $\mathbb{T}^2$ denotes the two dimensional torus and $r_+ = (2Ml^2)^{\frac{1}{3}}$ is the event horizon for this space. This metric was first stated in the papers \cite{metric1,metric2}.\\
In these coordinates $(t,r,x,y)$ we make the usual periodic identification for the torus 
\[
x\sim x+p_1,\hspace{2pt} y\sim y+p_2.
\]
As the toric domain is compact integrating over it shall prove no problem and these periods will typically not appear explicitly. There appears to be no obstacle in extending our results to the planar black hole (without identification of $x, y$), provided one assumes some decay for the field in the $x, y$ directions.

For later asymptotic analysis it is also helpful to know the expression of the cometric
\[
g^{-1} = -\left( \frac{2Ml^4}{r^5}+\frac{l^2}{r^2}\right)\pa_t^2 +\frac{4Ml^2}{r^3}\pa_t\pa_r + \left(\frac{-2M}{r}+\frac{r^2}{l^2}\right) \pa_r^2 + \frac{1}{r^2}(\pa_x^2+\pa_y^2). 
\]
\subsection{Hypersurfaces and Measures}
In this paper we will make extensive use of the divergence theorem. 
With this in mind it is helpful to establish a few hypersurfaces and measures.\\
The volume element for this space is
\[
d^4Vol = r^2drdtdxdy = r^2d\eta.
\]
We introduce the following spacetime slab
\[
\mathcal{M}_{[T_1,T_2]} = \mathbb{R}_{T_1\le t \le T_2} \times \mathbb{R}_{r\ge r_+}\times \mathbb{T}^2.
\]
We now define the hypersurfaces we will need
\begin{itemize}
\item $\Sigma_t$ the hypersurface of constant $t$. This surface has future directed unit normal given by
\[
\begin{split}
n &= \sqrt{-g^{tt}}\pa_t - \frac{g^{rt}}{\sqrt{-g^{tt}}}\pa_r\\
 &= \sqrt{\frac{2l^4M}{r^5}+\f{l^2}{r^2}}\pa_t - \frac{2lM}{\sqrt{2l^2Mr+r^4}}\pa_r,
 \end{split}
\]
and induced surface measure
\[
\begin{split}
dS_{\Sigma_t}&=\sqrt{-g^{tt}}r^2drdxdy\\
& = l\sqrt{\left( \frac{2l^2M}{r} + r^2\right) }drdxdy.
\end{split}
\]
A simple calculation shows that $\Sigma_t$ is a regular spacelike hypersurface up to and including the horizon. The following notation will also be useful
\[
\Sigma_t^{[R_1,R_2]} =\Sigma_t\cap \{R_1\le r\le R_2\}.
\]
\item $\Sigma_r$ the hypersurface of constant $r$. This surface has unit normal given by 
\[
\begin{split}
m &= \sqrt{g^{rr}}\pa_r - \frac{g^{tr}}{{\sqrt{g^{rr}}}}\pa_t\\
&= \frac{2l^3M}{\sqrt{-2l^2Mr^5+r^8}}\pa_t + \sqrt{\left( -\frac{2M}{r}+\f{r^2}{l^2}\right) }\pa_r,
\end{split}
\]
and induced surface measure
\[
\begin{split}
dS_{\Sigma_r}&= \sqrt{g^{rr}}r^2dtdxdy\\
&= \sqrt{-2Mr^3 + \frac{r^6}{l^2}} drdxdy.
\end{split}
\]
Notice that $m$ becomes singular and $dS_{\Sigma_r}$ degenerates as we approach the horizon. The combination $m^\mu dS_{\Sigma_r}$ is well behaved and gives the appropriate normal volume element to the surface. We again will make use of the notation
\[
\Sigma_r^{[T_1,T_2]} =\Sigma_r\cap \{T_1\le t\le T_2\}.
\]
\item The surface $T^2_{t,r}$ denotes the intersection of $\Sigma_t$ and $\Sigma_r$. It has induced surface measure
\[
dS_{T^2_{r,t}}= r^2 dxdy.
\]
\item We will also denote the hypersurface of the event horizon by
\[
\mh=\{r=r_+\},
\]
and define null infinity $\mi$ formally as
\[
\mi = \{r=\infty\}.
\]
\end{itemize} 
While working in this spacetime it is helpful to initially restrict to a finite region of the $t$ and $r$ coordinates then take limits to recover the full space. With this in mind we define
\[
B = \{(t,r,x,y)\in [T_1,T_2]\times[R_1,R_2]\times T^2\}.
\]

\begin{figure}[h]
\centering
\includegraphics{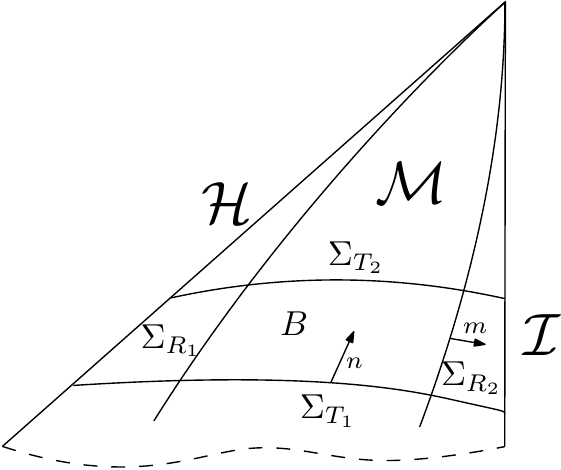}
\caption{Penrose diagram of the spacetime}
\end{figure}
\noindent We now state the divergence theorem for this setting. If we take $J^\mu$ to be a $C^1$ vector field on $\mathcal{M}$ then
\begin{equation}\label{DT}
\begin{split}
\int_{B} - \gr_\mu J^\mu d^4Vol &= \int_{{\Sigma^{[R_1,R_2]}_{T_2}}}J_\mu n^\mu dS_{\Sigma_{T_2}} - \int_{{\Sigma^{[R_1,R_2]}_{T_1}}}J_\mu n^\mu dS_{\Sigma_{T_1}} \\&+\int_{{\Sigma^{[T_1,T_2]}_{R_1}}}J_\mu m^\mu dS_{\Sigma_{R_1}} - \int_{{\Sigma^{[T_1,T_2]}_{R_2}}}J_\mu m^\mu dS_{\Sigma_{R_2}}.
\end{split}
\end{equation}
Providing the limits exist we can extend this to the spacetime slab $\mathcal{M}_{[T_1,T_2]}$ by sending $R_1\to r_+$ and $R_2 \to \infty$.\\ \\
To simplify our dealings with tangential terms we will denote the connection of the induced metric on tori of constant $t$ and $r$ by $\slashed{\gr}$.  We note that in our coordinate system for a function $u\in C^1(\mf)$ we have
\begin{equation*}
|\su|^2 = \frac{1}{r^2}\left(u_x^2+u_y^2 \right). 
\end{equation*}
\section{The Klein-Gordon Equation}
\subsection{Klein-Gordon Equation}
The Klein-Gordon for an asymptotically Anti de-Sitter spacetime is given as 
\begin{equation}\label{KGE}
\Box_g u + \frac{\alpha}{l^2}u = 0,
\end{equation}
where $\alpha < \frac{9}{4}$ obeys the Breitenlohner-Freedman bound \cite{bfb}. \\
Using the coordinate expansion 
\[
\Box_g u = \frac{1}{\sqrt{|g|}}\pa_\mu\left(\sqrt{|g|}g^{\mu\nu}\pa_\nu u \right), 
\]
for the wave operator, \eq{KGE} takes the form:
\[
 -\left( \frac{2ml^4}{r^5}+\frac{l^2}{r^2}\right) u_{tt}+\frac{1}{r^2}\pa_r\left(r^2\left( \frac{-2m}{r}+\frac{r^2}{l^2}\right) u_r\right) +\frac{4ml^2}{r^3}u_{rt} -\frac{2ml^2}{r^4}u_t + \frac{1}{r^2}\Delta_{(x,y)}u + \frac{\alpha}{l^2}u =0.
\]
We note that second order radial derivatives degenerate as $r\to r_+$. If we express this PDE in the form
\[
-u_{tt}+ B u_t+Lu=0,
\] 
for some spatial differential operators $B$, $L$ of degree one and two respectively, we find $L$ is not strongly elliptic at the horizon and that the standard energy methods are insufficient to prove boundedness of the full $H^1$ norm. We overcome this by exploiting the redshift effect for black holes. The details on how this can be done for general black holes can be found in \cite{DBH}. We must also confront the issue that standard energy fluxes of this PDE diverge as $r\to\infty$ this is fixed by a renormalization process, in particular by reformulating the problem in terms of the twisted derivative.  
\subsection{The Twisted derivative}
Due to the fact that asymptotically anti de-Sitter spaces do not admit a Cauchy hypersurface, in order to have a well-posed problem it is necessary to prescribe boundary data on $\mathcal{I}$. While the standard $\pa_t$ energy currents are sufficient to establish well-posedness for Dirichlet type data, for more general boundary conditions the solution has less radial decay and the standard energy diverges. This was resolved in \cite{warn1} for the range $\alpha \in (\frac{5}{4}, \frac{9}{4})$ by treating the well-posedness of the problem in asymptotically AdS spaces with the use of a re-normalisation scheme. The core idea is to reformulate the energies in terms of `twisted' derivatives
\[
\grt_\mu u = f\gr_\mu \left( f^{-1} u\right), 
\] 
for a `twisting' function $f>0$ which captures the decay of the field near $\mathcal{I}$. 

Defining $\grt^\dagger$ as the formal adjoint of $\grt$ with respect to the spacetime $L^2$ inner product
\[
\grt^\dagger_\mu u = -\frac{1}{f}\gr_\mu(f u),
\]
we note we can rewrite the Klein-Gordon equation in the form
\[
-\grt^\dagger_\mu \grt^\mu u - Vu = 0,
\]
where
\[
V = -\left( \frac{\gr_\mu\gr^\mu f}{f} + \frac{\alpha}{l^2}\right). 
\]

\subsection{Boundary conditions}
It will be helpful to define $\kappa > 0$ by $\alpha = \frac{9}{4} - \kappa^2$. We say that $u \in C^1(\mf, \mathbb{R})$ obeys Dirichlet, Neumann or Robin boundary conditions if the following hold
\begin{itemize}
\item Dirichlet, $\kappa >0$ and
\[
r^{\frac{3}{2}-\kappa}u \to 0, \text{ as } r \to \infty, 
\]
\item Neumann, $\kappa \in (0, 1)$ and
\[
r^{\frac{5}{2}+\kappa} \grt_ru \to 0, \text{ as } r \to \infty, 
\]
\item Robin, $\kappa \in (0, 1)$ and
\[
r^{\frac{5}{2}+\kappa} \grt_ru + \beta r^{\frac{3}{2}-\kappa}u \to 0, \text{ as } r \to \infty ,
\]
where $\beta \in C^\infty(\mathcal{I})$.\\
We will throughout this paper assume as part of our definition of Robin boundary conditions that $\pa_t\beta=0$ and $\beta \ge 0$.
\end{itemize}
\subsection{Well-posedness and asymptotics}
In this section we state the well-posedness results and the asymptotic behaviour of the solution as found in \cite{warn1}.\\
Firstly let $\Sigma$ be a smooth spacelike hypersurface which extends to $\mathcal{I}$ and meets it orthogonally\footnote{with respect to any conformal regularisation of the boundary.}. We
let $n_\Sigma$ be the future directed unit normal of $\Sigma$ and define
\[
\hat{n}_\Sigma = rn_\Sigma,
\]
then let $\mathcal{D}^+(\Sigma)$ denote the future Cauchy development of $\Sigma$ together with the portion of $\mathcal{I}$ lying in the future of $\Sigma$.\\
We choose our twisting function $f$ such that $fr^{\frac{3}{2}-\kappa}= 1+\mo(r^{-2})$ as $r \to \infty$ in order to define the norms
\[
\begin{split}
\|u\|_{\Lu^2(\Sigma)}^2 &= \int_\Sigma \frac{u^2}{r} dS_\Sigma,\\
\|u\|_{\Hu^1(\Sigma, \kappa)}^2 &= \int_\Sigma\left( |\grt u|^2 + \frac{u^2}{r^2}\right) rdS_\Sigma,
\end{split}
\]
and the space $\Hu_0^1(\Sigma,\kappa)$ as the completion of the smooth functions supported away from $\mi$. Different choices of twisting function $f$ satisfying the same asymptotic condition give rise to equivalent norms.
\begin{Theorem}[Well-posedness and asymptotics]
$ $\newline
\begin{itemize}
\item
Let $u_0 \in \Hu_0^1(\Sigma,\kappa),\hspace{2pt} u_1 \in \Lu^2(\Sigma)$. Then there exists a unique $u$ such that $u|_\Sigma = u_0, \hspace{2pt} \hat{n}_\Sigma u|_{\Sigma} = u_1$ which weakly solves
\[
\Box_gu + \frac{1}{l^2}\left(\frac{9}{4} - \kappa^2 \right) u = 0,
\]
in $\mathcal{D}^+(\Sigma)$ with Dirichlet boundary conditions on $\mi$. If $\mathcal{S}$ is any space like surface in $\mathcal{D}^+(\Sigma)$ meeting $\mi$ orthogonally then $u|_{\mathcal{S}} \in \Hu_0^1(\mathcal{S},\kappa) , \hspace{2pt} \hat{n}_\Sigma u|_{\mathcal{S}} \in \Lu^2(\mathcal{S})$.
\item
Let $u_0 \in \Hu^1(\Sigma,\kappa),\hspace{2pt} u_1 \in \Lu^2(\Sigma)$ and $0<\kappa<1$. Then there exists a unique $u$ such that $u|_\Sigma = u_0, \hspace{2pt} \hat{n}_\Sigma u|_{\Sigma} = u_1$ which weakly solves
\[
\Box_gu + \frac{1}{l^2}\left(\frac{9}{4} - \kappa^2 \right) u = 0,
\]
in $\mathcal{D}^+(\Sigma)$ with Neumann or Robin boundary conditions (for given $\beta$) on $\mi$. If $\mathcal{S}$ is any space like surface in $\mathcal{D}^+(\Sigma)$ meeting $\mi$ orthogonally then $u|_{\mathcal{S}} \in \Hu^1(\mathcal{S},\kappa),$ $\hat{n}_\Sigma u|_{\mathcal{S}} \in \Lu^2(\mathcal{S})$.
\end{itemize}
If the initial conditions satisfy stronger regularity and asymptotic conditions, then \\$u|_{\mathcal{S}} \in H_{\text{loc.}}^k(\mathcal{S}), $ $ \hat{n}_{\mathcal{S}}u|_{\mathcal{S}}\in H_{\text{loc.}}^{k-1}(\mathcal{S})$ for any integer $k\ge 2$ and we obtain an asymptotic expansion
\[
 u = \frac{1}{r^{\frac{3}{2}-\kappa}}\left(u^- + \mo(r^{-1-\kappa}) \right) +\frac{1}{r^{\frac{3}{2}+\kappa}}\left(u^+ + \mo(r^{\kappa-1}) \right).  
\]
Where the functions $u^- \in H^{k-1}(\mi), \hspace{2pt} u^+ \in H^{k-2}(\mi)$ satisfy
\[
\begin{split}
 u^- &= 0 \text{ if } u \text{ satisfies Dirichlet data}, \\
 u^+ &= 0  \text{ if } u \text{ satisfies Neumann data}, \\
 2\kappa u^+ - \beta u^- &= 0 \text{ if } u \text{ satisfies Robin data}.
 \end{split}
\]
\end{Theorem}
\begin{rem}
We will for the remainder of this paper assume our solutions are smooth, and admit asymptotic expansions to all orders. (Such solutions can be constructed from sufficiently smooth initial data.) This assumption can later be removed by a density argument. We state the asymptotics for the solution of the Dirichlet, Neumann and Robin problems, when $\kappa \in (0, 1)$:\\
Dirichlet:
\[
\begin{split}
u = \O{r^{-\frac{3}{2} - \kappa}},&\hspace{6pt}
\gr_tu =  \O{r^{-\frac{3}{2} - \kappa}},\\
|\su| =  \O{r^{-\frac{5}{2} - \kappa}},&\hspace{6pt}
\grt_ru = \O{r^{-\frac{5}{2} - \kappa}}, \\
\end{split}
\]
Neumann/Robin:
\begin{equation}\label{DRN}
\begin{split}
u = \O{r^{-\frac{3}{2} + \kappa}},&\hspace{6pt}
\gr_tu =  \O{r^{-\frac{3}{2} + \kappa}},\\
|\su| =  \O{r^{-\frac{5}{2} + \kappa}},&\hspace{6pt}
\grt_ru = \O{r^{-\frac{5}{2} + \kappa}}. \\
\end{split}
\end{equation}

\end{rem}
\subsection{The Twisted Energy Momentum Tensor} \label{TEMT}
When deriving energy estimates for wave equations on the exterior of black holes one typically considers the energy momentum tensor. By choosing suitable vector fields as multipliers and commutators one can prove boundedness and decay of solutions (for a discussion of the history of this approach in the black hole setting see \cite{DBH}). In the case of AdS spaces with Neumann and Robin boundary data one finds that for the standard energy estimate obtained by using the vector field $T=\pa_t$ as a multiplier, the energy is no longer finite. This is resolved by the introduction of the twisted energy momentum tensor defined as
\[
\temt_{\mu\nu}[u] = \grt_\mu u \grt_\nu u - \hf g_{\mu\nu}\left(\grt_\sigma u \grt^\sigma u + Vu^2 \right), 
\]       
where again:
\[
V = -\left( \frac{\gr_\mu\gr^\mu f}{f} + \frac{\alpha}{l^2}\right). 
\]
Unfortunately the divergence of the twisted energy-momentum tensor (in contrast to the untwisted version) is no longer vanishing for solutions to \eqref{KGE}. However, it does have some useful properties
\begin{Lemma}[taken from \cite{warn2}]\ \\
\begin{itemize}
\item For $\psi \in C^2(\mf)$
\[
\gr_\mu\temt^\mu{}_\nu[\psi] = \left( -\grt^\dagger_\mu \grt^\mu \psi - V\psi\right)\grt_\nu \psi + \tilde{S}_\nu[\psi], 
\]
where
\[
\tilde{S}_\nu[\psi] = \frac{\grt^\dagger_\nu(fV)}{2f}\psi^2 + \frac{\grt^\dagger_\nu f}{2f}\grt_\sigma \psi\grt^\sigma \psi.
\]
\item For $u$ a solution to the PDE and $X$ a smooth vector field. Defining
\[
\tilde{J}^X_\mu[u]=\temt_{\mu\nu}[u]X^\nu, \hspace{10pt} \tilde{K}^X[u] = {}^X\pi_{\mu\nu}\temt^{\mu\nu}[u] + X^\nu\tilde{S}_\nu[u],
\]
where ${}^X\pi_{\mu\nu}$ is the deformation tensor 
\[
{}^X\pi_{\mu\nu} =  \hf \left( \gr_\mu X_\nu + \gr_\nu X_\mu\right)= \hf\left( \mathcal{L}_X g\right)_{\mu\nu}, 
\] 
we have
\[
\gr^\mu\tilde{J}^X_\mu[u] = \tilde{K}^X[u].
\]
\item If $f$ is chosen such that $V\ge 0$ then $\temt_{\mu\nu}$ satisfies the dominant energy condition, i.e. for a future direct causal vector field $Y$ we have that $-\temt^\mu{}_{\nu}Y^\nu$ is future directed and timelike. If $Y$ is timelike then $\temt_{\mu\nu}Y^\mu Y^\nu$ controls all the first order derivatives of $u$. 
\end{itemize}
\end{Lemma}
As we have that $\tilde{S}_\mu$ and thus $\tilde{K}^X[u]$ depend only on the $1-$jet of $u$ we get that $\tilde{J}^X_\mu[u]$ is a compatible current in the sense of Christodoulou \cite{CCC}.  \\
Importantly if we have that $Z$ is a Killing field that preserves $f$, i.e. $\mathcal{L}_{Z}(f) = 0$ then $\tilde{J}^Z_\mu[u]$ is a conserved current.\\ \\
We remark that the renormalization by the twisted derivative encompasses the holographic renormalization of \cite{BS1} and \cite{BS2}. 
\section{Energy Decay}
With the notions of the previous sections in mind we are now in a position to prove an energy decay result for the Klein-Gordon equation in this setting. This follows the same method as in \cite{warn2} which deals with the spherical case.\\
Firstly we choose our twisting function. Experimentation suggests the following function 
\[
f(r) = r^{-\frac{3}{2} + \kappa}.
\]
This choice gives 
\[
V(r) = \frac{(3-2\kappa)^2M}{2r^3} > 0, 
\]
and we can easily verify that for the timelike vector field
\[
T= \pa_t,
\]
we have
\[
\mathcal{L}_T(f) =0,
\]
so that
\[
\gr^\mu \tilde{J}^T_\mu = 0.
\]
We now integrate over $B$ to get the following identity
\[
\int_{{\Sigma^{[R_1,R_2]}_{T_2}}}\tilde{J}^T_\mu n^\mu dS_{\Sigma_{T_2}} - \int_{{\Sigma^{[R_1,R_2]}_{T_1}}}\tilde{J}^T_\mu n^\mu dS_{\Sigma_{T_1}}=  \int_{{\Sigma^{[T_1,T_2]}_{R_2}}}\tilde{J}^T_\mu m^\mu dS_{\Sigma_{R_2}}-\int_{{\Sigma^{[T_1,T_2]}_{R_1}}}\tilde{J}^T_\mu m^\mu dS_{\Sigma_{R_1}},
\]
Through a long but straight forward calculation one arrives at 
\[
\begin{split}
\mathcal{E}_t(u;[R_1,R_2]) :&=\int_{{\Sigma^{[R_1,R_2]}_{t}}}\tilde{J}^T_\mu n^\mu dS_{\Sigma_{t}}\\&= \hf\int_{{\Sigma^{[R_1,R_2]}_{t}}}\left( -g^{tt}(\gr_t u)^2 + g^{rr}(\grt_ru)^2 + |\su|^2 +V(r)u^2\right)r^2drdxdy,\\
\mathcal{F}_r(u;[T_1,T_2]) :&= \int_{{\Sigma^{[T_1,T_2]}_{r}}}\tilde{J}^T_\mu m^\mu dS_{\Sigma_{r}} \\&=\int_{{\Sigma^{[T_1,T_2]}_{r}}}\left( g^{rt}(\gr_tu)^2+g^{rr}(\gr_t u)(\grt_r u) \right)r^2 dtdxdy.
\end{split}
\]
We now wish to take the limits $R_1 \to r_+$ and $R_2 \to \infty$ in the fluxes so that energy is defined across the whole exterior of the black hole. Approaching the event horizon we get
\[
\lim_{R_1 \to r_+} \mathcal{F}_{R_1}(u;[T_1,T_2]) = \int_{\mathcal{H}_{[T_1,T_2]}}g^{rt}(\gr_t u)^2 r_+^2 dtdxdy =: F(u;[T_1,T_2]),
\]  
we observe that $g^{rt}$ is positive on the horizon making this flux a positive quantity. \\
As for the contribution on $\mi$ we can quickly see from the asymptotics that for Dirichlet and Neumann data that
\[
\lim_{R_2 \to \infty} \mathcal{F}_{R_2}(u;[T_1,T_2])= 0,
\]   
however for time independent Robin data
\[
\begin{split}
\lim_{R_2 \to \infty} \mathcal{F}_{R_2}(u;[T_1,T_2])&= - \lim_{r \to \infty} \int_{{\Sigma^{[T_1,T_2]}_{r}}} \frac{\beta}{2l^2}\gr_t((r^{\frac{3}{2}-\kappa}u)^2)dtdxdy \\
& = \frac{1}{2l^2}\int_{T^2_{T_1,\infty}}(r^{\frac{3}{2}-\kappa}u)^2\beta dxdy - \frac{1}{2l^2}\int_{T^2_{T_2,\infty}}(r^{\frac{3}{2}-\kappa}u)^2\beta dxdy,
\end{split}
\] 
where we understand the terms in the integral at $r=\infty$ to mean $\lim_{r\to\infty}  r^{\frac{3}{2}-\kappa}u$.
We may thus define the renormalised energy for a function $u$ as 
\[
\begin{split}
E_{t}[u] &= \hf\int_{{\Sigma_{t}}}\left( -g^{tt}(\gr_t u)^2 + g^{rr}(\grt_ru)^2 + |\su|^2 +\frac{(3-2\kappa)^2M}{2r^3}u^2\right)r^2drdxdy\\ &+ \frac{1}{2l^2}\int_{T^2_{t,\infty}}(r^{\frac{3}{2}-\kappa}u)^2\beta dxdy,
\end{split}
\]
where we include the latter term for Robin data.
This energy is positive definite and finite for our boundary data. It also satisfies the useful identity
\begin{equation}\label{DEI}
E_{T_2}[u] = E_{T_1}[u] - F(u;[T_1,T_2]),
\end{equation}  
so that $E_{t}[u]$ is a strictly non-increasing function of $t$.
Using the redshift effect \cite{DBH}, we can remove the degeneracy at the horizon to establish:
\begin{Theorem}\label{thm1}
Suppose $u$ is a solution to \eqref{KGE} satisfying suitable (Dirichlet, Neumann or Robin) boundary conditions at infinity. Define the non-degenerate renormalised energy density $\mathcal{E}[u]$ by
\ben{renormalized energy}
\mathcal{E}[u] :=  \frac{1}{r} u^2  + r^4 (\grt_ru)^2 + (\gr_tu)^2 + r^2|\su|^2.
\een
Then there exists a constant $C=C(M, l, \kappa)>0$ such that  for any $T_1<T_2$ we have:
\begin{align*}
\int_{\Sigma_{T_2}}\mathcal{E}[u]dr dx dy  \leq C\int_{\Sigma_{T_1}}\mathcal{E}[u]dr dx dy 
\end{align*}
\end{Theorem}
\begin{rem}
If enough regularity is assumed on the initial data of \eq{KGE} one may extend this result: by applying $T$ and  the red shift vector field as commutators, together with elliptic estimates and a Sobolev embedding we can extract pointwise boundedness. In this setting one can prove a result similar to that found in \cite{warn2}.  
\end{rem}
\section{Decay Rates}
Now that we have established energy decay for solutions to $\eqref{KGE}$ we turn to the problem of establishing integrated decay. Due to limitations of our energy current we show the existence of a $\kappa^*$ with $\frac{3}{4}<\kappa^*<1,$ such that our results hold for $\kappa \in \left( 0,\kappa^*\right)$. We believe the results hold for all $\kappa \in (0,1)$ however the energy current will be far more complex.  We will examine all the previously listed boundary conditions. We remark that this range of $\kappa$ includes the conformally coupled case $(\kappa = \hf)$. The methods we will use were first established by Morawetz \cite{mora2} and \cite{mora1} for the obstacle problem for waves and also seen in \cite{DRS}, \cite{DBH} and \cite{labso} for the Schwarzschild black hole. The core idea is to repeat a similar argument as in the energy decay proof, only rather than using $\pa_t$ as a multiplier we examine vector fields of the form $h(r)\pa_r$ for some radial function $h$. 
\subsection{The Morawetz Estimate} 
In this section, we shall establish the following integrated decay estimate
\begin{Theorem}\label{thm2}[Precise Version]
Suppose $u$ is a smooth solution to \eqref{KGE} with $\kappa \in (0,\kappa^*),$ satisfying Dirichlet, Neumann or Robin boundary conditions at infinity. Then for any $T_1<T_2$ we have:
\begin{align} \nonumber
\int_{\mf_{[T_1,T_2]}} &\left(\frac{1}{r} u^2  + r^4 (\grt_ru)^2 +\frac{1}{r^3}(\gr_tu)^2 + \frac{1}{r}|\su|^2\right) d\eta   \leq C \int_{\Sigma_{T_1}}\mathcal{E}[u] dr dx dy\label{ME1}
\end{align}
with $C=C(M, l, \kappa)>0$.
\end{Theorem}
To prove this we make use of the following two lemmas

\begin{Lemma}[Divergence of the Modified Energy Current]\label{DMEC}
Let $u$ be a solution to \eqref{KGE}  with Dirichlet, Neumann or Robin boundary conditions. For an energy current defined as
\begin{equation*}
\tilde{J}^\mu[u] = \tilde{\mathbb{T}}_\nu{}^\mu[u] X^\nu + w_1(r) u\grt^\mu u + w_2(r) u^2 X^\mu,
\end{equation*}
where 
\begin{equation*}
X = r\pa_r,
\end{equation*}
we then have
\begin{equation} \label{Gdiv}
\begin{split}
\gr_\mu\tilde{J}^\mu[u] &= \left(2  r w_2(r)+\frac{r^2 w_1'(r)}{l^2}-\frac{2 M w_1'(r)}{r}\right)u\grt_r u+\frac{2 l^2 M  w_1'(r)}{r^3}u\gr_tu\\&+ \left( r w_2'(r)+2 \kappa 
   w_2(r)+\frac{(3-2\kappa)^3  M}{4r^3}+\frac{2 (3-2\kappa)^2 M w_1(r)}{4 r^3}\right)u^2\\& + 4 \left(\frac{(2-\kappa)  l^2 M}{r^3}+\frac{ l^2 M w_1(r)}{r^3}\right)\grt_ru \gr_tu+\left((1-\kappa)+{w_1(r)} \right)|\su|^2 \\&+ 
   \left(-\frac{ \kappa r^2}{ l^2}-\frac{(3-2\kappa)  M}{r}+\frac{r^2 w_1(r)}{l^2}-\frac{2 M w_1(r)}{r}\right)(\grt_ru)^2\\&+ \left(-\frac{(5-2\kappa)  l^4 M}{r^5}-\frac{2 l^4 M w_1(r)}{r^5}-\frac{l^2}{r^2}\left( 
   (1-\kappa)+w_1(r)\right) \right)(\gr_tu)^2.
\end{split}
\end{equation}
\end{Lemma}
The proof of this may be found in the technical lemmas section of the paper.
\begin{Lemma} [Bounded Integrated Divergence] \label{MECB}
For a current defined as
\begin{equation}
\tilde{J}^\mu[u] = \tilde{\mathbb{T}}_\nu{}^\mu[u] X^\nu + w_1(r) u\grt^\mu u + w_2(r) u^2 X^\mu,
\end{equation}	
where
\begin{itemize}
\item $u$ solves \eqref{KGE} with Dirichlet, Neumann or Robin data. 
\item $X = r\pa_r$,
\item $w_1 = -k_1 + f(r),$\\
  with  $k_1> 0$ and $f\in \mo(r^{-3}),$
\item $w_2 = \frac{k_2}{r^3},$\\
with $
0\le k_2< \frac{(3-2\kappa)^2}{4}M.$
\end{itemize}
We have that
\begin{align*}
\int_{\mf_{[T_1,T_2]}} -\gr_\mu \tilde{J}^\mu[u] d^4Vol \le C \int_{\Sigma_{T_1}}\mathcal{E}[u]dr dx dy
\end{align*}
for a constant $C$ independent of $T_1$ and $T_2$.
\end{Lemma}
The proof of this may be found in the technical lemmas section of the paper.
\begin{proof}[Proof of Theorem \ref{thm2}]
We split the estimate into two parts as inspection of \eqref{Gdiv} reveals that the coefficient of the tangential derivatives appears with the opposite sign to the higher order terms of the time derivative's coefficient. This is problematic when trying to get a signed divergence. So we initially cancel this term off to control a positive quantity. This establishes part of the Morawetz estimate and is then used to control a divergence that includes poorly signed terms of lower order but crucially correctly signed tangential terms.     

We define our first current as 
\[
\tilde{J}_1^\mu[u] = \tilde{\mathbb{T}}_\nu{}^\mu[u] X^\nu + (\kappa-1) u\grt^\mu u + \frac{(3 - 2\kappa)M}{2r^3}\left(\hf +\frac{\epsilon}{2} \right)  u^2 X^\mu,
\]
with 
\[
X = r\pa_r,
\]
as our multiplier. It can be easily checked that this current satisfies Lemma \ref{MECB} providing $\epsilon<2(1-\kappa)$ and by using the divergence formula from Lemma \ref{DMEC} we get
\[
\begin{split}
-\gr_\mu\tilde{J}_1^\mu[u] \cdot r^2 &= \frac{3 l^4 M }{r^3}(\gr_tu)^2-\frac{4 l^2 M }{r}\grt_ru \gr _tu+\frac{(3-2\kappa)^2 M\epsilon
   }{4 r}  u^2\\&-\hf(3-2\kappa)(1+\epsilon) M u\grt_ru+\left(\frac{r^4}{l^2}+ M r\right) (\grt_ru)^2.
\end{split}
\]
which we can re-write as:
\ben{first estimate}
\begin{split}
-\gr_\mu\tilde{J}_1^\mu[u] \cdot r^2 &= \frac{3 l^4 M }{r^3}(\gr_tu)^2-\frac{4 l^2 M }{r}\grt_ru \gr _tu+\frac{r^4}{l^2} (\grt_ru)^2\\&+ \frac{(3-2\kappa)^2 M\epsilon
   }{4 r}  u^2-\hf(3-2\kappa)(1+\epsilon) M u\grt_ru+ M r (\grt_ru)^2.
\end{split}
\een
Note that our choice of $w_1$ ensures that there is no term involving $\su$.
We now seek to bound a positive quantity from above by this divergence. First, note that
\be
\abs{\frac{4 l^2 M }{r}\grt_ru \gr _tu} \leq \frac{4 l^4 M }{ \delta r^3}(\gr_tu)^2 + \delta  M r (\grt_ru)^2,
\ee
for some $\delta>0$ to be later determined. We thus have
\begin{align*}
 \frac{3 l^4 M }{r^3}(\gr_tu)^2-\frac{4 l^2 M }{r}\grt_ru \gr _tu+ \frac{r^4}{l^2} (\grt_ru)^2 &\geq  \frac{ l^4 M }{ r^3}\left(3-\frac{4}{\delta} \right) (\gr_tu)^2 +   \left( \frac{r^4}{l^2} -\delta Mr \right)(\grt_ru)^2.
\end{align*}
This deals with the first line of \eq{first estimate}. Now, for the second line
\[
\begin{split}
 &\frac{(3-2\kappa)^2 M\epsilon
   }{4 r}  u^2-\hf(3-2\kappa)(1+\epsilon) M u\grt_ru+ M r (\grt_ru)^2\\
 &\ge \frac{M}{r}\left(\epsilon\left(\frac{3-2\kappa}{2}\right)^2u^2 - \frac{\beta}{2}\left(\frac{3-2\kappa}{2}\right)^2u^2- \frac{(1+\epsilon)^2}{2\beta}r^2(\grt_ru)^2 + r^2(\grt_ru)^2   \right),   
\end{split}
\]
where $\beta>0$ is a constant from Young's inequality that we shall later determine.
Factoring the above we deduce
\[
\begin{split}
 &\frac{(3-2\kappa)^2 M\epsilon
   }{4 r}  u^2-\hf(3-2\kappa)(1+\epsilon) M u\grt_ru+ M r (\grt_ru)^2\\
 &\ge \frac{M}{r}\left(\left( \epsilon-\frac{\beta}{2}\right) \left(\frac{3-2\kappa}{2}\right)^2u^2 +\left( 1- \frac{(1+\epsilon)^2}{2\beta}\right) r^2(\grt_ru)^2  \right),   
\end{split}
\]
now combining everything we have
\[
\begin{split}
-\gr_\mu\tilde{J}_1^\mu[u] \cdot r^2 \ge& \frac{ l^4 M }{ r^3}\left(3-\frac{4}{\delta} \right)(\gr_tu)^2 +  \left( \frac{r^4}{l^2} -Mr\left(\delta +\frac{(1+\epsilon)^2}{2\beta}-1\right)   \right)(\grt_ru)^2\\&+ \frac{M}{r}\left( \epsilon-\frac{\beta}{2}\right) \left(\frac{3-2\kappa}{2}\right)^2u^2,
\end{split}
\]
so providing that 
\begin{equation}\label{eq1}
\beta < 2\epsilon, \qquad \delta> \frac{4}{3},
\end{equation}
and
\begin{equation}\label{eq2}
\delta +\frac{(1+\epsilon)^2}{2\beta}-1 <2,
\end{equation}
we have a bounded positive quantity.\\ We now show for which range of $\epsilon$ this can hold and then convert to the mass range this method works for.\\Firstly note \eqref{eq2} is equivalent to
\[
\delta <3-\frac{(1+\epsilon)^2}{2\beta},
\]
but for consistency we then need
\[
\frac{4}{3}<3-\frac{(1+\epsilon)^2}{2\beta},
\]
which is equivalent to
\[
\frac{3}{10}(1+\epsilon)^2 < \beta,
\]
again, for consistency we then also need
\[
\frac{3}{10}(1+\epsilon)^2 <2\epsilon,
\]
which is equivalent to
\[
3\epsilon^2-14\epsilon+3<0,
\]
which holds for the range
\[
\epsilon\in\left( \frac{1}{3}(7-2\sqrt{10}),\frac{1}{3}(7+2\sqrt{10})\right), 
\]
now we also have the restraint that
\[
\epsilon<2-2\kappa,
\]
that is
\[
\kappa < 1-\frac{\epsilon}{2},
\]
which is extremized for
\[
\kappa < \frac{1}{6} \left( 2 \sqrt{10}-1\right) \approx 0.887,
\]
thus providing an upper bound for the $\kappa$ ranges we can currently prove for.
If we also make the choice
\[
\epsilon = 1-\kappa,
\]
we can easily then show the positivity and boundedness for the range
\[
\kappa \in \left( 0,\frac{1}{6} \left( 2 \sqrt{10}-1\right)\right). 
\]
So for $\kappa$ in this range we may find all the constants to deduce that
\be
-\gr_\mu\tilde{J}_1^\mu[u] \cdot r^2 \geq c\left(  \frac{ 1 }{ r^3}(\gr_tu)^2 + r^4 (\grt_ru)^2+  \frac{1}{r}u^2\right), 
\ee
for some $c = c(M, l, \kappa)>0$. We remark that we have used $r\ge (2Ml^2)^{\frac{1}{3}}$.

We now invoke Lemma \ref{MECB} to deduce
\begin{equation} \label{NRMor1}
\int_{\mf_{[T_1,T_2]}}\left(  \frac{ 1 }{ r^3}(\gr_tu)^2 + r^4 (\grt_ru)^2+  \frac{1}{r}u^2\right) dtdrdxdy  \leq C \int_{\Sigma_{T_1}}\mathcal{E}[u] dr dx dy.
\end{equation}

Now that we have control of $u$, $\gr_t u$ and $\grt_r u$ in an integrated sense, we can return to establish an estimate involving $\su$. To obtain this, we consider the current:
\begin{equation*}
\tilde{J}_2^\mu =- \left( \frac{1}{r^3}+(1-\kappa)\right)  u\grt^\mu u .
\end{equation*}
The divergence of this current is readily calculated to give:
\begin{align*}
-\gr_\mu\tilde{J}_2^\mu[u] \cdot r^2 &=  \left(\frac{6 M}{r^3}-\frac{3}{l^2}\right)u\grt_ru-\frac{6 l^2 M }{r^5}u\gr_tu-\frac{4 l^2 M \left(r^3-1\right)}{r^4}\grt_ru \gr_tu\\&+ \left(\frac{r^4+r}{l^2}+M \left(r-\frac{2}{r^2}\right)\right)(\grt_ru)^2+\frac{ \left(l^4 M \left(3
   r^3-2\right)-l^2 r^3\right)}{r^6}(\gr_tu)^2\\&-\frac{(3-2 \kappa )^2 M \left(r^3-2\right)}{4 r^4}u^2+\frac{1}{r}|\su|^2
\end{align*}
Examining the coefficients of the terms involving $u$, $\gr_t u$ and $\grt_r$ we see that we already control all of these terms with appropriate weights. In particular, it is clear that we can find a $C>0$ such that:
\be
 \frac{1}{r} \abs{\su}^2 \leq -\gr_\mu\tilde{J}_2^\mu[u] \cdot r^2  + C\left(  \frac{ 1 }{ r^3}(\gr_tu)^2 + r^4 (\grt_ru)^2+  \frac{1}{r}u^2\right).
\ee
Integrating this estimate, applying Lemma \ref{MECB} and estimate \eq{NRMor1}, we finally conclude
\begin{align*}
\int_{\mf_{[T_1,T_2]}} \left(\frac{1}{r} u^2  + r^4 (\grt_ru)^2 +\frac{1}{r^3}(\gr_tu)^2 + \frac{1}{r}|\su|^2\right) d\eta \leq C \int_{\Sigma_{T_1}}\mathcal{E}[u] dr dx dy
\end{align*}
for some $C>0$.
\end{proof}

\subsection{Integrated Decay Estimate without weight loss}

We can now restate our result regarding the integrated decay with no loss in the radial weights:
\begin{Theorem}\label{thm3} 
Suppose $u$ is a smooth solution to \eqref{KGE}  with $\kappa \in (0,\kappa^*),$ satisfying Dirichlet, Neumann or Robin boundary conditions at infinity. Let $\mathcal{E}[u]$ be the renormalised energy density of the field $u$, as in \eq{renormalized energy}. Then 
\begin{equation} \label{ME2}
\int_{\mf_{[T_1,T_2]}} \mathcal{E}[u] d\eta \leq C \int_{\Sigma_{T_1}}\left ( \mathcal{E}[u] +  \mathcal{E}[u_t]\right)dr dx dy
\end{equation}
for some $C = C(M, l, \kappa)>0$.
\end{Theorem}   

In order to prove Theorem \ref{thm3}, we shall require the following result:
\begin{Lemma}[Hardy Estimate] \label{HEL}
Let $\phi:[r_+, \infty) \to \R$ be a smooth function such that $lim_{r\to\infty}r^\hf \phi=0$. Then the following inequality holds
\begin{equation}\label{HE}
\int_{r_+}^\infty \phi^2 dr \le C\left( \int_{r_+}^\infty \frac{\phi^2}{r}dr+\int_R^\infty(\grt_r\phi)^2 r^2 dr \right). 
\end{equation}
Where $C=C(R, r_+)>0$.
\end{Lemma}
The proof of this may be found in the technical lemmas section of the paper.
\begin{rem}
For solutions to \eqref{KGE} we have from our earlier asymptotics for Robin data that $r^\hf u \in \mo\left(r^{-1+\kappa} \right)$ so the above Hardy estimate holds.
\end{rem}
\begin{proof}[Proof of Theorem \ref{thm3}]
In order to improve the weights in Theorem \ref{thm2} we first apply the Hardy estimate of Lemma \ref{HEL} to establish that
\ben{improved est}
\int_{\mf_{[T_1,T_2]}} u^2d\eta \leq C \int_{\Sigma_{T_1}}\mathcal{E}[u] dr dx dy,
\een
holds under the same conditions for $u$ as in Theorem \ref{thm2}.
Next we exploit some of the spacetime's symmetry. We note that 
\begin{equation}
\left[ \pa_t, \Box_g + \frac{\alpha}{l^2}\right]  = 0,
\end{equation}
this tells us that $\left(\gr_tu\right)$ satisfies \eqref{KGE} (provided the initial data is regular enough) and so \eq{improved est} applies. This yields an estimate of the form
\begin{equation*}
\int_{\mf_{[T_1,T_2]}} (\gr_tu)^2 d\eta \le C \int_{\Sigma_{T_1}}\mathcal{E}[u_t] dr dx dy.
\end{equation*}
we may then recombine with the estimate in Theorem \ref{thm2} and see
\begin{equation}\label{MENT}
\int_{\mf_{[T_1,T_2]}} (\gr_tu)^2+r^4(\grt_ru)^2 +u^2 d\eta \le C \int_{\Sigma_{T_1}}\left ( \mathcal{E}[u] +  \mathcal{E}[u_t]\right)dr dx dy
\end{equation}
it now only remains to recover the tangential derivatives. We will proceed with a more robust estimate here but refer the reader to the remark after this proof if they want a faster but less frugal route.\\
We now define the current 
\[
\tilde{J}_3^\mu[u] = \tilde{\mathbb{T}}_\nu{}^\mu[u] X^\nu -(2-\kappa) u\grt^\mu u 
\]
where
\[
X = r\pa_r
\]
from examining \eqref{Gdiv} we can see that this current won't pick up any cross terms. More explicitly we get 
\[
-\gr_\mu\tilde{J}_3^\mu \cdot r^2 =\left( \frac{l^4 M }{r^3}-l^2\right) (\gr_tu)^2+\left(\frac{2 r^4}{l^2} -Mr\right)(\grt_ru)^2+\frac{M(3-2\kappa)^2 }{4r}u^2+r^2|\su|^2
\]
this current clearly satisfies the conditions of lemma \ref{MECB} and as we control all the non-tangential terms we have that 
\begin{equation*}
\int_{\mf_{[T_1,T_2]}} r^2|\su|^2 d\eta  \le C \int_{\Sigma_{T_1}}\left ( \mathcal{E}[u] +  \mathcal{E}[u_t]\right)dr dx dy
\end{equation*}  
which we combine to get
\begin{equation*}
\begin{split}
\int_{\mf_{[T_1,T_2]}} (\gr_tu)^2+r^4 (\grt_ru)^2 +u^2+r^2|\su|^2 d\eta \le C \int_{\Sigma_{T_1}}\left ( \mathcal{E}[u] +  \mathcal{E}[u_t]\right)dr dx dy
\end{split}
\end{equation*}
and finally
\begin{equation}\label{IDE}
\int_{\mf_{[T_1,T_2]}} \mathcal{E}[u] d\eta \leq C \int_{\Sigma_{T_1}}\left ( \mathcal{E}[u] +  \mathcal{E}[u_t]\right)dr dx dy.
\end{equation}
\end{proof}
\begin{rem}
If we aren't concerned about the robustness of this result we can also exploit
\begin{equation}
\left[ \pa_x, \Box_g+\frac{\alpha}{l^2}\right]  = 0,
\end{equation} 
\begin{equation}
\left[ \pa_y, \Box_g+\frac{\alpha}{l^2}\right]  = 0.
\end{equation} 
These would then manifest as $\gr_xu$ and $\gr_y u$ control. This would give us an estimate of the form
\be
\int_{\mf_{[T_1,T_2]}} \mathcal{E}[u] d\eta \leq C \int_{\Sigma_{T_1}}\left ( \mathcal{E}[u] +  \mathcal{E}[u_t]+  \mathcal{E}[u_x] +  \mathcal{E}[u_y]\right)dr dx dy
\ee 
and this would be enough for a decay estimate. With nonlinear applications in mind, avoiding explicit use of the toric symmetry is preferable.
\end{rem}

In order to establish a decay statement, we require the following straightforward corollary of Gronwall's Lemma:
\begin{Lemma}Let $k>0$, and suppose that $f\in C^1([T_1, \infty))$ satisfies \label{newgronlem}
\ben{newgron}
f'(t) \leq -\varkappa f(t) + \frac{A}{\left(1+t-T_1\right)^k},
\een
then there exists $C=C(k, \varkappa)>0$ such that
\be
f(t) \leq f(T_1) e^{-\varkappa (t-T_1)}+  \frac{C A}{\left(1+t-T_1\right)^k}.
\ee
\end{Lemma}
The proof of this may be found in the technical lemmas section of the paper.\\

We also require the following quantitative version of the redshift effect, which can be found in Theorem 3.8 of \cite{Warnick:2013hba}.
\begin{Lemma}\label{RSE}
There exists a modified energy $\mathbb{E}_t[u]$, the \emph{redshift energy} such that
\begin{itemize}
\item $\mathbb{E}_t[u]$  is equivalent to the non-degenerate energy at time $t$. That is, for any smooth $u$ satisfying appropriate boundary conditions we have:
\be
C^{-1} \int_{\Sigma_{t}} \mathcal{E}[u] dr dx dy \leq  \mathbb{E}_t[u] \leq C \int_{\Sigma_{t}} \mathcal{E}[u] dr dx dy
\ee
for some $C>0$.
\item If $u$ solves \eq{KGE} subject to Dirichlet, Neumann or Robin boundary conditions, then $\mathbb{E}_t[u]$ satisfies 
\ben{redshift}
\frac{d}{dt} \mathbb{E}_t[u]  \leq - \varkappa\mathbb{E}_t[u] + C E_t[u] 
\een
for some $\varkappa>0$.
\end{itemize}
\end{Lemma}

\begin{Theorem}\label{thm3.5}
Suppose $u$ is a solution to \eqref{KGE} with $\kappa \in (0,\kappa^*),$ with Dirichlet, Neumann or Robin data. Then 
\begin{equation}
 \int_{\Sigma_{t}} \mathcal{E}[u] dr dx dy \le \frac{C}{(1+t)^n} \sum_{k=0}^n   \int_{\Sigma_{0}} \mathcal{E}[\pa_t^k u] dr dx dy
\end{equation}
for some $C = C(n, M, l,\kappa)>0$.
\end{Theorem}
\begin{proof}
We can easily see that 
\[
E_t[u] \le C\int_{\Sigma_t}\mathcal{E}[u]drdxdy
\]
as the quantities are equivalent norms for $u$ away from the horizon. Integrating in time we get
\be
\int_{T_1}^{T_2} E_s[u] ds \leq C \int_{\mf_{[T_1, T_2]}} \mathcal{E}[u] d\eta  
\ee
so by theorem \ref{thm3} we have that
\be
\int_{T_1}^{T_2} E_s[u] ds \leq C \int_{\Sigma_{T_1}}\left ( \mathcal{E}[u] +  \mathcal{E}[u_t]\right)dr dx dy.
\ee
Since $E_t[u]$ is non-increasing, a slight adaptation of \cite{warn3}, Lemma 5.8 immediately implies that for any $t>T_1$ we have:
\be
E_t[u]   \leq \frac{C}{1+t-T_1} \int_{\Sigma_{T_1}}\left ( \mathcal{E}[u] +  \mathcal{E}[u_t]\right)dr dx dy.
\ee
This gives us decay of the degenerate energy. To establish decay of the full non-degenerate energy we apply Lemma \ref{newgronlem} to \eq{redshift}, and we deduce that
\begin{align*}
\mathbb{E}_t[u] &\leq \mathbb{E}_{T_1}[u]e^{-\varkappa(t-T_1)} + \frac{C}{1+t-T_1} \int_{\Sigma_{T_1}}\left ( \mathcal{E}[u] +  \mathcal{E}[u_t]\right)dr dx dy. \\
& \leq C e^{-\varkappa(t-T_1)} \int_{\Sigma_{T_1}} \mathcal{E}[u] dr dx dy + \frac{C}{1+t-T_1} \int_{\Sigma_{T_1}}\left ( \mathcal{E}[u] +  \mathcal{E}[u_t]\right)dr dx dy \\
&\leq  \frac{C}{1+t-T_1} \int_{\Sigma_{T_1}}\left ( \mathcal{E}[u] +  \mathcal{E}[u_t]\right)dr dx dy
\end{align*}
Finally, we obtain:
\be
 \int_{\Sigma_{t}} \mathcal{E}[u] dr dx dy  \leq \frac{C}{1+t-T_1} \int_{\Sigma_{T_1}}\left ( \mathcal{E}[u] +  \mathcal{E}[u_t]\right)dr dx dy.
\ee
Which gives the result for $n=1$ on setting $T_1 = 0$. For higher $n$, the induction argument follows precisely as in  \cite{warn3}, Lemma 5.8.
\end{proof}

\subsection{Gaussian Beam and Derivative Loss}
We now have an integrated decay estimate but with derivative loss. While we have not quantified precisely how much derivative loss is required we will show that it is necessary. That is 
\begin{Theorem}\label{thm4} 
There exists no constant $C>0$, independent of $T$, such that the estimate 
\begin{equation}\label{GBE}
\int_{\mf_{[0,T]}} \mathcal{E}[u]  d\eta \le C \int_{\Sigma_{0}} \mathcal{E}[u] dr dx dy
\end{equation}
holds for all smooth solutions $u$ of \eqref{KGE}.
\end{Theorem}  
The proof of this comes from the Gaussian beam construction as seen in \cite{GBS}. The core idea is to show that there exists null geodesics that can remain outside the event horizon for arbitrary lengths of coordinate time. We then can construct approximate solutions to \eqref{KGE} (Gaussian beams) supported in a tubular neighbourhood of these geodesics such that they lose arbitrary small amounts of energy along them and remain close to true solutions in the energy norm. As we can find solutions which lose arbitrary small amounts of energy for any fixed time interval we cannot have an estimate of the form in Theorem \ref{thm4}.   
\begin{Lemma}\label{TNG}
For any given $T>0$, there exists a null geodesic $\gamma$ such that $Im(\gamma)\subset \mf_{[0,T]}\cap \left \{\frac{3}{2} r_+<r<R(T)\right\}$ for some large $R(T)$. Furthermore $\gamma$ is a smooth embedding.
\end{Lemma}
The proof of this may be found in the technical lemmas section of the paper.\\

We now use a modified theorem from \cite{GBS}. These modifications are due to the fact that we are not working on a globally hyperbolic manifold but due to our geodesic being a smooth embedding an analogous result holds.
\begin{Theorem}[Gaussian Beam]\label{GBT}
Let $(\mf,g)$ be a time oriented Lorentzian manifold with time function $t$, foliated by the level sets $\Sigma_\tau = \{t=\tau\}$. Furthermore, let $\gamma$ be a smooth geodesic embedding that intersects $\Sigma_0$ and $N$ a timelike, future directed vector field.\\
For any neighbourhood $\mathcal{N}$ of $\gamma$, for any $T>0$ with $\Sigma_T\cap Im(\gamma)\ne \emptyset$ there exists a Gaussian beam $u_\lambda$ of the form $u_\lambda(x) = a_{\mathcal{N}}(x)e^{i\lambda\phi(x)}$ with the following properties
\[
\norm{\Box_gu_\lambda}{L^2(\mf_{[0,T]})}\le C(T),
\]
where the constant $C(T)$ depends on $a_{\mathcal{N}},\phi$ and $T$, but not on $\lambda$,
\[
E_t[{u_\lambda}]\to \infty \hspace{10pt} \text{for } \lambda\to\infty,
\]
and
\[
u_\lambda \text{ is supported in } \mathcal{N},
\] 
provided we have on $\mf_{[0,T]}\cap J^+(\mathcal{N}\cap\Sigma_0),$
\begin{equation*}
\begin{split}
&\frac{1}{|n_{\Sigma_\tau}(t)|}\le C,\hspace{10pt} g(N,N) \le -c <0, \hspace{10pt} -g(N,n_{\Sigma_\tau}(t))\le C,\\
&{} \hspace{130pt}\text{and}\\
&|g(\gr_{n_{\Sigma_\tau}}N,n_{\Sigma_\tau})|, |g(\gr_{n_{\Sigma_\tau}} N,e_i)|, |g(\gr_{e_i} N,e_j)|\le C, \hspace{10pt} \text{for $1\le i,j \le 3$},
\end{split}
\end{equation*}
where $c$ and $C$ are positive constants and $\{n_{\Sigma_\tau},e_1,e_2,e_3\}$ is an orthonormal frame.
\end{Theorem}
In our case the vector field $N$ is $\pa_t$ and the time function is simply the $t$ coordinate. With the null geodesic from lemma \ref{TNG} all the conditions can be easily seen to hold provided we bound $\mathcal{N}$ away from $\mi$.\\
We will also need the following result about $\phi$ from \cite{GBS} (2.14)
\begin{Lemma}\label{GBL0}
$\Im(\phi|_{\gamma})= \Im(\gr\phi|_{\gamma})=0$ and $\Im(\gr\gr\phi|_\gamma)$ is positive definite on a $3$-dimensional subspace transversal to $\dot{\gamma}$.
\end{Lemma}
Combining the fact that $a_{\mathcal{N}}$ is independent of $\lambda$ with \eqref{GBL0} we have that the $L^2$ norm of our Gaussian beam is independent of $\lambda$. We collect this observation in a lemma,
\begin{Lemma}
	Let $u_\lambda$ be the function constructed in \ref{GBT} then there exists a constant $C(T)>0$ independent of $\lambda$ such that the following bound holds:
	\begin{equation}
	\norm{u_\lambda}{L^2(\mf_{[0,T]})} \le C(T).
	\end{equation} 
\end{Lemma} 
\begin{Lemma}\label{GBL1}
For all $\epsilon>0$ there exists a solution $v$ of \eqref{KGE} for all boundary data types, initial data supported away from the horizon, with $E_0[v]=1$ and a Gaussian beam $\tilde{u_\lambda}$ such that
\[
|E_t[v]-E_t[{\tilde{u_\lambda}}]| < \epsilon, \hspace{10pt} \forall\hspace{2pt} 0\le t\le T.
\]
\begin{proof}
Firstly construct $u_\lambda$ from theorem \ref{GBT} using the geodesic in lemma \ref{TNG} ensuring that $\mathcal{N}$ is bounded away from $\mi$ and $\mathcal{H}$, and then define 
\begin{equation*}
\tilde{u}_\lambda := \frac{u_\lambda}{\sqrt{E_0[{u_\lambda}]}},
\end{equation*}
so we have (by the triangle inequality)
\[
\norm{\Box_g\tilde{u}_\lambda+\frac{\alpha}{l^2}\tilde{u}_\lambda}{L^2(\mf_{[0,T]})} \to 0,
\] 
as $\lambda \to \infty$.\\
Now set $v$ to be the solution to
\begin{equation*}
\begin{split}
\Box_g v + \frac{\alpha}{l^2}v &= 0,\\
v|_{\Sigma_0} &= \tilde{u}_\lambda|_{\Sigma_0},\\
n_{\Sigma_0}v|_{\Sigma_0} &= n_{\Sigma_0}\tilde{u}_\lambda|_{\Sigma_0},
\end{split}
\end{equation*}
with Dirichlet, Robin or Neumann boundary conditions.\\
We now apply the same energy estimates as in section 4 to ${\tilde{u}_\lambda}$, (we remark that as $\tilde{u}_\lambda$ doesn't solve \eqref{KGE} this introduces an inhomogeneity in the bulk integral) yielding the inequality
\[
E_t[{\tilde{u}_\lambda}] \le E_0[{\tilde{u}_\lambda}] +\norm{\left( \Box_g\tilde{u}_\lambda+\frac{\alpha}{l^2}\tilde{u}_\lambda\right)\gr_t{\tilde{u}_\lambda} }{L^2(\mf_{[0,T]})}^2 ,
\]
after a simple application of Cauchy-Schwartz we find that
\[
E_t[{\tilde{u}_\lambda}] \le C(T)\left(E_0[{\tilde{u}_\lambda}] +T\cdot\sup_{t\in [0,T]}\left( E_t[\tilde{u}_\lambda]\right) \norm{\Box_g\tilde{u}_\lambda+\frac{\alpha}{l^2}\tilde{u}_\lambda}{L^2(\mf_{[0,T]})} \right),
\]
for all $0\le t \le T$. (Note that $\mathcal{N}$ being bounded away from infinity controls the flux there).\\
Taking supremums on both sides and absorbing an application of Young's inequality we get 
\[
E_t[{\tilde{u}_\lambda}] \le C(T)\left(E_0[{\tilde{u}_\lambda}] +\norm{\Box_g\tilde{u}_\lambda+\frac{\alpha}{l^2}\tilde{u}_\lambda}{L^2(\mf_{[0,T]})} \right).
\] 
Applying this inequality to the difference $v-\tilde{u_\lambda}$, gives us the result
\begin{equation*}
|E_t[{v-\tilde{u_\lambda}}]| \le C(T) \norm{\Box_g\tilde{u}_\lambda+\frac{\alpha}{l^2}\tilde{u}_\lambda}{L^2(\mf_{[0,T]})} \hspace{10pt} \forall\hspace{2pt} 0\le t\le T,
\end{equation*} 
as they agree on $\Sigma_0$. So for fixed $\epsilon$ we simply choose $\lambda_0$ large enough and set $\tilde{u} := \tilde{u}_{\lambda_0}$ and $ v = v_{\lambda_0}$.
\end{proof}
\end{Lemma}
We now invoke Theorem 2.36 from \cite{GBS} which tells us that the Gaussian beam energy is localised around geodesic energy. That is
\begin{Lemma}\label{GBL2}
For all $\epsilon >0$ there exists a neighbourhood $\mathcal{N}_0$ of $\mathcal{N}$ such that
\begin{equation*}
|E_t[{\tilde{u}_{\lambda}}|_{\mathcal{N}_0}] -\left( -g(T,\dot\gamma)|_{Im(\gamma)\cap\Sigma_t}\right) |<\epsilon,
\end{equation*}
for all $0\le t \le T$.
\end{Lemma}
\begin{Lemma}\label{GBL3}
Let $T >0$ then for all $\epsilon>0$ there exists a solution $u$ of \eqref{KGE} whose initial data is supported away from the horizon, that can satisfy any choice of the discussed boundary conditions and with $E_0[u]=1$ such that
\begin{equation*}
E_t[u] \ge 1-\epsilon,
\end{equation*}
for all $0\le t\le T.$
\end{Lemma}
\begin{proof}
As $T$ is Killing we have that $-g(T,\dot\gamma)|_{Im(\gamma)\cap\Sigma_\tau}$ is constant. We may choose this constant to be $1$ when we solve for $\gamma$ (see the proof of lemma \ref{TNG} in the technical lemmas section for details). Applying the triangle inequality to the results in lemmas \ref{GBL1} and \ref{GBL2} yields the result.
\end{proof} 
\begin{proof}[Proof of Theorem \ref{thm4}]
We proceed by contradiction and assume there exists a constant $C$ independent of $T$ and $u$ such that 
\begin{equation}\label{degee}
\int_{0}^{T} E_t[u]dt \le C E_0[u],
\end{equation}
as the energy is decreasing we have
\begin{equation*}
\int_{0}^{T} E_T[u]dt \le \int_{0}^{T} E_t[u]dt \le C E_0[u],
\end{equation*}
this gives us the estimate
\begin{equation*}
TE_T[u] \le CE_0[u],
\end{equation*}
now let $T=2C$  and construct $u$ from lemma \ref{GBL3} where we choose $\epsilon= \f{1}{4}$ we then have
\begin{equation*}
\f{3}{2}C \le C,
\end{equation*}
which is clearly a contradiction. Now that we have established that \eqref{degee} cannot hold we may extend the result to the non-degenerate energy.\\
If $u$ is supported away from the horizon then there exists $C>0$ such that
\[
\int_{\Sigma_0}\mathcal{E}[u] \le C E_0[u],
\]
and clearly
\[
\int_0^T E_t[u]dt\le C\int_{\mf_{[0,T]}} \mathcal{E}[u]  d\eta.
\]
Assume by contradiction that the statement of Theorem \ref{thm4} holds true. Then the above inequalities entail equation \eqref{degee} contradicting the first part of the proof.
\end{proof}

\section{Technical Lemmas}
\subsection{Proof of Lemma \ref{HEL}}
\begin{proof}
Firstly define a cut-off function

\[ 
\chi(r) = \left\{ 
  \begin{array}{l l}
    0 & \quad \text{if $r \le R$},\\
    1 & \quad \text{if $r \ge 2R$}, \\
    \text{Smooth}& \quad \text{if $r\in [R,2R].$} 
  \end{array} \right.
\]
With the property $\chi'(r) \le \frac{C}{r}$ for some $C>0$ and monotone. \\
This can be achieved by a suitable bump function.\\
Now we write 
\[
  \|u\|_{L^2} =\|\chi u + (1-\chi)u\|_{L^2},
\]
After applying the triangle inequality we estimate the terms separately
\[
\begin{split}
\|(1-\chi) u \|^2_{L^2} = \int_{r_+}^\infty (1-\chi)^2 u^2 dr \le \int_{r_+}^{2R}r \cdot \frac{u^2}{r}dr \le 2R \int_{r_+}^\infty\frac{u^2}{r}dr. 
\end{split}
\]
And for the other term
\[
\begin{split}
\|\cu\|^2_{L^2(r_+,\infty)}&=\int_{r_+}^\infty (\cu)^2dr  = \int_{R}^\infty (\cu)^2 dr = \int_R^\infty \left( \cu r^{\frac{3}{2}-\kappa}\right) ^2\pa_r \left( \frac{r^{-2+2\kappa}}{2\kappa-2}\right)  dr \\
&= \underbrace{\left[ \frac{1}{2\kappa-2}(\cu)^2r \right] _R^\infty}_{=0} + \frac{2}{2-2\kappa}\int_R^\infty \cu (\grt_r\cu)rdr \le \frac{1}{1-\kappa}\|\cu\|_{L^2(r_+,\infty)}\|r\grt_r\cu\|_{L^2(R,\infty)}.
\end{split}
\]
Now looking at
\[
\begin{split}
\|r\grt_r\cu\|_{L^2(R,\infty)}^2 &= \int_{R}^\infty r^2(\grt_r\cu)^2 dr \\&\le C \int_{R}^\infty r^2(\chi \grt_ru)^2 + r^2(u \pa_r\chi)^2 dr \\&\le C\left( \int_{R}^\infty (\grt_ru)^2 r^2 dr + \int_{r_+}^\infty \frac{u^2}{r}dr\right), 
\end{split}
\]
Combining all these estimates yields \eqref{HE}.
\end{proof}
\subsection{Proof of Lemma \ref{DMEC}}
\begin{proof}\ \\
We recall that
\begin{equation*}
\gr^\mu\left( \temt_{\mu\nu}X^\nu\right)  = {}^X\pi_{\mu\nu}\temt^{\mu\nu} + X^\nu\tilde{S}_\nu.
\end{equation*}
We start with the latter term. A quick computation shows us
\begin{equation*}
\tilde{S}_r = \frac{1}{r}\left((3-\kappa)V(r)u^2 + \frac{(3-2\kappa)}{2}\grt_\mu u\grt^\mu u\right), 
\end{equation*}
and thus
\begin{equation*}
X^\nu\tilde{S}_\nu =(3-\kappa)V(r)u^2 + \frac{(3-2\kappa)}{2}\grt_\mu u\grt^\mu u
\end{equation*}
We now deal with the deformation tensor term. We can factor the metric out of the deformation tensor 
\begin{equation*}
{}^X\pi = g - \frac{3M}{r}dt^2 - \frac{8l^2M}{r^3} dtdr - \frac{l^2(5l^2M+r^3)}{r^5}dr^2.
\end{equation*}
Contracting $g$ into $\temt$ we get
\begin{equation*}
g_{\mu\nu}\temt^{\mu\nu} = - \grt_\mu u\grt^\mu u-2V(r)u^2,
\end{equation*}
then combining with the $\tilde{S}$ terms we get
\begin{equation*}
g_{\mu\nu}\temt^{\mu\nu} + X^\nu\tilde{S}_\nu = (1-\kappa)V(r)u^2 + \frac{(1-2\kappa)}{2}g^{\mu\nu}\grt_\mu u\grt_\nu u
\end{equation*}
Contracting the remaining deformation terms we get 
\begin{equation*}
\frac{6 l^2 M}{r^3}\grt_ru\gr_tu-\frac{l^2 \left(8 l^2 M+r^3\right)}{2 r^5}(\gr_tu)^2-\frac{ 4 l^2 M+r^3}{2
   l^2 r}(\grt_ru)^2+\frac{M(3-2\kappa)^2}{4r^3}u^2+\hf|\su|^2.
\end{equation*}
So for the first term we conclude
\begin{equation*}
\begin{split}
\gr^\mu\left( \temt_{\mu\nu}X^\nu\right)  &= -\frac{(2 \kappa -3)^3 M}{4 r^3}u^2+\left( \frac{(2 \kappa -3) M}{r}-\frac{\kappa  r^2}{l^2}\right) (\grt_r u)^2\\
&+ (1-\kappa)|\su|^2 -\frac{4 (\kappa -2) l^2 M}{r^3}\grt_r u\gr_tu +\left(\frac{(\kappa -1) l^2}{r^2} +\frac{(2 \kappa -5) l^4 M}{r^5} \right)(\gr_tu)^2 
  \end{split}
\end{equation*}
Now we compute the divergence of the second term. To do this we write
\begin{equation*}
\gr_\mu u = \left( \grt_\mu u - u\grt_\mu 1\right), 
\end{equation*}
then we can compute
\begin{equation*}
\begin{split}
\gr_\mu \left(w_1(r) u \grt^\mu u \right) &= \gr_\mu(w_1u)\grt^\mu u + w_1u \gr_\mu\grt^\mu u,
\end{split}
\end{equation*}
now we explore part of the second term
\begin{equation*}
\begin{split}
\gr_\mu\grt^\mu u  & = \Box_g u + \gr_\mu\left(  u\grt^\mu 1\right) \\
&=\Box_g u+ u\gr_\mu\grt^\mu1 + \gr_\mu u \grt^\mu 1\\
&= \Box_g u + u \gr_\mu \left( g^{\mu\nu}\grt_\nu 1\right)  + \gr_\mu u g^{\mu\nu}\grt_\nu 1 \\
& = \Box_g u + u\gr_\mu\left( \left( \frac{3}{2}-\kappa\right) \frac{g^{\mu r}}{r} \right) + \left( \frac{3}{2}-\kappa\right)\frac{g^{\mu r}}{r}\gr_\mu u \\
&= \Box_g u + \left( \frac{3}{2}-\kappa\right)u\gr_\mu \left( \frac{g^{\mu r}}{r}\right)  + \left( \frac{3}{2}-\kappa\right)\frac{g^{\mu r}}{r}\left( \grt_\mu u - u\grt_\mu 1\right)  \\
&= \Box_g u   + \left( \frac{3}{2}-\kappa\right)u\gr_\mu\left( \frac{g^{\mu r}}{r} \right) - \left( \frac{3}{2}-\kappa\right)\frac{u}{r}g^{\mu r}\grt_\mu 1  + \frac{\left( \frac{3}{2}-\kappa\right)}{r} \grt^ru\\
&= \Box_g u  + \left( \frac{3}{2}-\kappa\right)u\frac{1}{r^2}\gr_r\left( r^2g^{rr} \right) - \left( \frac{3}{2}-\kappa\right)^2\frac{g^{rr}}{r^2}u  + \frac{\left( \frac{3}{2}-\kappa\right)}{r} \grt^ru \\
&= -\frac{\alpha}{l^2}u + \frac{3\left( \frac{3}{2}-\kappa\right)}{l^2} u + \frac{2M\left( \frac{3}{2}-\kappa\right)^2}{r^3}u - \frac{\left( \frac{3}{2}-\kappa\right)^2}{l^2}u+ \frac{1}{r} \grt^ru\\
&= \frac{2M\left( \frac{3}{2}-\kappa\right)^2}{r^3}u + \frac{1}{r} \grt^ru.
\end{split}
\end{equation*}
Recombining we get
\begin{equation*}
\begin{split}
\gr_\mu(w_1u)\grt^\mu u + w_1u \gr_\mu\grt^\mu u &= \left( \gr_r w_1\right)u\grt^ru + w_1\gr_\mu u \grt^\mu u + w_1u\left(\frac{2M}{r^3} u + \frac{1}{r} \grt^ru\right) \\
&= \left( \gr_r w_1\right)u\grt^ru + w_1\grt_\mu u\grt^\mu u - w_1u\grt_\mu 1\grt^\mu u\\&+ w_1u\left(\frac{2M}{r^3} u + \frac{1}{r} \grt^ru\right) \\
&= g^{\mu r}\left( \gr_r w_1\right)u\grt_\mu u + w_1g^{\mu\nu}\grt_\mu u\grt_\nu u+ w_1\frac{2M}{r^3} u^2. 
\end{split}
\end{equation*}
Explicitly this is
\begin{equation*}
\begin{split}
&\frac{2 l^2 M  w_1'(r)}{r^3}u\gr_tu+ \left(\frac{r^2 w_1'(r)}{l^2}-\frac{2 M w_1'(r)}{r}\right)u\grt_ru+\frac{4 l^2
   M  w_1(r)}{r^3}\grt_ru \gr_tu\\&+ \left(\frac{r^2 w_1(r)}{l^2}-\frac{2 M w_1(r)}{r}\right)(\grt_ru)^2+ \left(-\frac{2 l^4 M w_1(r)}{r^5}-\frac{l^2 w_1(r)}{r^2}\right)(\gr_tu)^2\\&+w_1(r)\frac{2M}{r^3}u^2+ w_1(r)|\su|^2.
\end{split}
\end{equation*} 
For the last term
\begin{equation*}
\begin {split}
\gr_\mu(w_2(r)u^2X^\mu) &= \gr_\mu(w_2u^2)X^\mu + w_2u^2\gr_\mu X^\mu\\
&= \gr_\mu(w_2)u^2X^\mu + 2w_2 u \gr_\mu u X^\mu + 3 w_2u^2\\
& = \left( r w'_2(r) + 3w_2(r) \right) u^2 + 2rw_2(r)u\gr_r u \\
& = \left( r w'_2(r) + 2\kappa w_2(r) \right) u^2 +2rw_2(r)u\grt_r u 
\end{split}
\end{equation*}
Combining all of these results and tidying up the algebra we may now fully express the divergence for the current as
\begin{equation}
\begin{split}
\gr_\mu\tilde{J}^\mu[u] &= \left(2 r w_2(r)+\frac{r^2 w_1'(r)}{l^2}-\frac{2 M w_1'(r)}{r}\right)u\grt_r u+\frac{2 l^2 M  w_1'(r)}{r^3}u\gr_tu\\&+ \left( r w_2'(r)+2 \kappa 
   w_2(r)+\frac{(3-2\kappa)^3  M}{4r^3}+\frac{2 (3-2\kappa)^2 M w_1(r)}{4 r^3}\right)u^2\\& + 4 \left(\frac{(2-\kappa)  l^2 M}{r^3}+\frac{ l^2 M w_1(r)}{r^3}\right)\grt_ru \gr_tu+\left((1-\kappa)+{w_1(r)} \right)|\su|^2 \\&+ 
   \left(-\frac{ \kappa r^2}{ l^2}-\frac{(3-2\kappa)  M}{r}+\frac{r^2 w_1(r)}{l^2}-\frac{2 M w_1(r)}{r}\right)(\grt_ru)^2\\&+ \left(-\frac{(5-2\kappa)  l^4 M}{r^5}-\frac{2 l^4 M w_1(r)}{r^5}-\frac{l^2}{r^2}\left( 
   (1-\kappa)+w_1(r)\right) \right)(\gr_tu)^2.
\end{split}
\end{equation}
\end{proof}
\subsection{Proof of Lemma \ref{MECB}}
\begin{proof}
To prove this we apply the divergence theorem \eqref{DT} to the current and prove the fluxes across the hypersurfaces are bounded by a constant multiple of the initial energy. \\
For the surfaces of constant $r$ we compute that:
\begin{equation}\label{rflux}
\begin{split}
&\int_{\Sigma_r^{[T_1,T_2]}}\tilde{J}_\mu m^\mu dS_{\Sigma_r} = \\
&\int_{\Sigma_r^{[T_1,T_2]}}\left( r^3 w_2(r)- \frac{(3-2\kappa)^2}{4}M\right)u^2+\frac{r^2}{2l^2}\left(r^3 - 2Ml^2 \right) (\grt_ru)^2+  \left(\frac{l^4
	M}{r^2}+\frac{l^2 r}{2}\right)(\gr_tu)^2\\&-\frac{1}{2}  r^3|\su|^2+ w_1(r) \frac{r}{l^2}\left(r^3-2 Ml^2\right)u\grt_ru+\frac{2 l^2 M  w_1(r)}{r}u\gr_tu \hspace{2pt} dtdxdy.
\end{split}
\end{equation}
We first investigate the terms on $\mi$. Using the asymptotic analysis for Robin data as seen in \eqref{DRN} we can quickly see most of these terms are converging to $0$. The limit thus reduces to the study of
\begin{equation*}
\begin{split}
&\lim_{r\to\infty}\frac{r^2}{2l^2}\left(r^3 - 2Ml^2 \right) (\grt_ru)^2 + w_1(r) \frac{r}{l^2}\left(r^3-2 Ml^2\right)u\grt_ru\\
=&\lim_{r\to\infty} \frac{1}{2l^2r^{2\kappa}}(r^{\frac{5}{2}+\kappa} \grt_ru)^2  -\frac{k_1}{l^2}\left( r^{\frac{3}{2}-\kappa}u\right)r^{\frac{5}{2}+\kappa}\grt_ru \\
=& \lim_{r\to\infty}  \frac{k_1\beta}{l^2}\left( r^{\frac{3}{2}-\kappa}u\right)^2 = \lim_{r\to\infty}\lambda \beta\left( r^{\frac{3}{2}-\kappa}u\right)^2,
\end{split}
\end{equation*}
where in the last line we have substituted in the limit for Robin data and $\lambda>0$ is some constant independent of $T$. In integral form this is
\begin{equation*}
\begin{split}
\lim_{r\to\infty}\int_{\Sigma_r^{[T_1,T_2]}}\tilde{J}_\mu m^\mu dS_{\Sigma_r} &=\lim_{r\to\infty}\int_{\Sigma_r^{[T_1,T_2]}}\lambda \beta\left( r^{\frac{3}{2}-\kappa}u\right)^2\\
&=\lambda\int_{T_1}^{T_2}\int_{T_{t,\infty}^2} \beta\left( r^{\frac{3}{2}-\kappa}u\right)^2dxdydt.
\end{split}
\end{equation*}
As for the contribution of the terms at the horizon $(r=r_+)$ we first bound \eqref{rflux} using Young's inequality
\begin{equation*}
\begin{split}
\lim_{r\to r_+}&\int_{\Sigma_r^{[T_1,T_2]}}\tilde{J}_\mu m^\mu dS_{\Sigma_r} \le \\
\lim_{r\to r_+}&\int_{\Sigma_r^{[T_1,T_2]}}\left( r^3 w_2(r)- \frac{(3-2\kappa)^2}{4}M+\epsilon\right)u^2+\frac{r^2}{2l^2}\left(r^3 - 2Ml^2 \right) (\grt_ru)^2\\&+  \left(\frac{l^4
	M}{r^2}+\frac{l^2 r}{2}+\frac{ l^4 M^2  w_1^2(r)}{\epsilon r^2}\right)(\gr_tu)^2-\frac{1}{2}  r^3|\su|^2+ w_1(r) \frac{r}{l^2}\left(r^3-2 Ml^2\right)u\grt_ru \hspace{2pt} dtdxdy\\
\le &\int_{\Sigma_{r_+}^{[T_1,T_2]}}\left( r^3 w_2(r_+)- \frac{(3-2\kappa)^2}{4}M+\epsilon \right)u^2+ \left(\frac{l^4
	M}{r_+^2}+\frac{l^2 r_+}{2}+\frac{ l^4 M^2  w_1^2(r_+)}{\epsilon r_+^2}\right)(\gr_tu)^2 \hspace{2pt} dtdxdy.\\
	\le&\int_{\Sigma_{r_+}^{[T_1,T_2]}}\left( k_2 - \frac{(3-2\kappa)^2}{4}M+\epsilon \right)u^2 + CF_u[T_1,T_2]
\end{split}
\end{equation*} 
For $k_2 < \frac{(3-2\kappa)^2}{4}M$ we can always find an $\epsilon>0$ such that the first bracket in the integrand is negative.  We may drop the negative terms and evaluate all the radial functions absorbing their values into a constant $C$. We thus conclude
\begin{equation*}
\begin{split}
&\lim_{r\to r_+}\int_{\Sigma_r^{[T_1,T_2]}}\tilde{J}_\mu m^\mu dS_{\Sigma_r} \le \int_{\Sigma_{r_+}^{[T_1,T_2]}}C(\gr_tu)^2\hspace{2pt} dtdxdy \le CF_u[T_1,T_2] \le C E_{T_1}[u].
\end{split}
\end{equation*}
We now have expressions for the $\mi$ and $\Sigma_{r_+}$ contributions.\\
As the energy is decreasing in time both surfaces of constant $t$ can be dealt with in one calculation. So for $T_1 \le t \le T_2$ we compute 
\begin{equation*}
\begin{split}
&\int_{\Sigma_t}\tilde{J}_\mu m^\mu dS_{\Sigma_t}  \\
=&\int_{\Sigma_t} -2  l^2 M (\grt_ru)^2+ \left(\frac{2 l^4 M}{r^2}+l^2 r\right) \grt_ru \gr_tu+w_1(r) \left(\frac{2 l^4 M}{r^3}+l^2\right)u\gr_tu\\&-\frac{2 l^2 M w_1(r)}{r}u\grt_ru \hspace{2pt} drdxdy\\
=&\int_{\Sigma_t} -2  l^2 M (\grt_ru)^2+ \left(\frac{2 l^4 M}{r^2}+l^2 r\right)\grt_ru \gr_tu-k_1l^2u\gr_tu \\
&+ \left(\frac{2 l^4 M(f(r)-k_1)}{r^3}+l^2f(r)\right)u\gr_tu-\frac{2 l^2 M w_1(r)}{r}u\grt_ru \hspace{2pt} drdxdy\\
= &\int_{\Sigma_t} -2 l^2 M (\grt_ru)^2+h_1(r)\grt_ru \gr_tu-k_1l^2u\gr_tu
+ h_2(r)u\gr_tu-h_3(r)u\grt_ru \hspace{2pt} drdxdy,
\end{split}
\end{equation*}
where $h_1(r) \in \mo(r)$, $h_2(r) \in \mo(r^{-1})$ and $h_3(r) \in \mo(r^{-1})$. With this in mind we now apply Young's inequality to the cross terms
\begin{equation*}
\begin{split}
&\int_{\Sigma_t}\tilde{J}_\mu m^\mu dS_{\Sigma_t}  \\
\le &\int_{\Sigma_t}\left(\frac{1}{2\epsilon_1} + \frac{k_1l^4}{2\epsilon_3}+\hf |h_2(r)| \right) (\gr_tu)^2+\left(\hf|h_2(r)|+\frac{1}{2\epsilon_2}|h_3(r)|\right) u^2 \hspace{2pt} drdxdy \\
+ &\int_{\Sigma_t}\left(\frac{\epsilon_2}{2}|h_3(r)|+\frac{\epsilon_1}{2}(h_1(r))^2-2l^2M \right)(\grt_ru)^2 + \frac{k_1^2l^4\epsilon_3}{2} u^2 \hspace{2pt} drdxdy.
\end{split}
\end{equation*}
Now $r$ is bounded away from zero and as the function coefficients are of a low enough order (in terms of $r$) we can easily bound the first integral by a constant multiple of the energy at time $T_1$ (the constant depends on the choice of the $\epsilon_i $'s). To deal however with the second integral we need to invoke the Hardy inequality
\begin{equation*}
\begin{split}
&\int_{\Sigma_t}\tilde{J}_\mu m^\mu dS_{\Sigma_t}  \\
\le &\int_{\Sigma_t}\left(\frac{1}{2\epsilon_1} + \frac{k_1l^4}{2\epsilon_3}+\hf |h_2(r)| \right) (\gr_tu)^2+\left(\hf|h_2(r)|+\frac{1}{2\epsilon_2}|h_3(r)|\right) u^2 \hspace{2pt} drdxdy\\&+\int_{\Sigma_t} C_2\frac{k_1^2l^4\epsilon_3}{2r}u^2 + \left(C_2\frac{k_1^2l^4\epsilon_3}{2}r^2+\frac{\epsilon_2}{2}|h_3(r)|+\frac{\epsilon_1}{2}(h_1(r))^2-2l^2M \right)(\grt_ru)^2 \hspace{2pt} drdxdy.
\end{split}
\end{equation*} 
We now note that we simply need to choose positive $\epsilon_1$, $\epsilon_2$ and $\epsilon_3$ such that
\begin{equation*}
C_2\frac{k_1^2l^4\epsilon_3}{2}r_+^2+\frac{\epsilon_2}{2}|h_3(r_+)|+\frac{\epsilon_1}{2}(h_1(r_+))^2<2l^2M, 
\end{equation*}
we may then find a $C>0$ independent of $T_1$ and $T_2$ such that
\begin{equation*}
\int_{\Sigma_t}\tilde{J}_\mu m^\mu dS_{\Sigma_t}  \le CE_{T_1}[u],
\end{equation*}
for $T_1\le t\le T_2$.\\
We now apply the divergence theorem to yield
\begin{equation*}
\int_{\mf_{[T_1,T_2]}} -\gr_\mu \tilde{J}^\mu[u] d^4Vol \le C_1E_{T_1}[u] + C_2E_{T_1}[u] + C_3E_{T_1}[u] - C\int_{T_1}^{T_2}\int_{T_{t,\infty}^2}\beta(ru)^2dxdydt,\\
\end{equation*}
which may be rewritten as
\begin{equation*}
\int_{T_1}^{T_2}\int_{T_{t,\infty}^2}\beta(ru)^2dxdydt+\int_{\mf_{[T_1,T_2]}} -\gr_\mu \tilde{J}^\mu[u] d^4Vol \le CE_{T_1}[u],    
\end{equation*}
thus proving the proposition.
\end{proof}
\subsection{Proof of Lemma \ref{newgronlem}}
\begin{proof}
We first re-write \eq{newgron} as:
\be
\frac{d}{dt} \left(f(t) e ^{\varkappa t} \right) \leq \frac{A e ^{\varkappa t}}{\left(1+t-T_1\right)^k}
\ee
so that
\begin{align*}
e^{\varkappa t} f(t) - e^{\varkappa T_1} f(T_1) & \leq A \int_{T_1}^t \frac{ e ^{\varkappa s}}{\left(1+s-T_1\right)^k} ds \\
&\leq A e^{\varkappa T_1} \int_{0}^{t-T_1} \frac{ e ^{\varkappa s'}}{\left(1+s'\right)^k} ds' \\
& \leq C A e^{\varkappa T_1}\left(  \frac{e ^{\varkappa (t-T_1)}}{\left(1+t-T_1\right)^k}  \right)
\end{align*}
Here, we use:
\begin{align*}
\int_{0}^{t} \frac{ e ^{\varkappa s'}}{\left(1+s'\right)^k} ds' &= \left[ \frac{e ^{\varkappa s'}}{\varkappa \left(1+s'\right)^k}  \right]_0^t + k \int_{0}^{t} \frac{  e ^{\varkappa s'}}{\varkappa \left(1+s'\right)^{k+1}} ds' \\
&\leq \frac{e ^{\varkappa t}}{\varkappa \left(1+t\right)^k} + k t \max_{s'\in (0, t) } \abs{ \frac{  e ^{\varkappa s'}}{\varkappa \left(1+s'\right)^{k+1}}} \\
&\leq C  \frac{e ^{\varkappa t}}{\left(1+t\right)^k}
\end{align*}
a simple re-arrangement gives \eq{newgron}.
\end{proof}\subsection{Proof of Lemma \ref{TNG}}
\begin{proof}
The easiest way to construct the geodesic is by the Hamiltonian method. We can quickly spot three integrals of motion arising from the Killing fields $\pa_t,\pa_x,\pa_y$ and from the fact that $\dot{\gamma}$ is null. More explicitly we define three constants $a,b,c$ and write the equations
\begin{equation*}
\begin{split}
a &= g(\gd,\pa_t),\\
b &= g(\gd,\pa_x),\\
c &= g(\gd,\pa_y),\\
0 &= g(\gd,\gd).
\end{split}
\end{equation*}
Working in co-ordinates and taking an affine parameter $\tau$ for $\gamma(t,r,x,y)$. We get the following geodesic equation
\begin{equation}\label{FDE}
\begin{split}
\td & = \frac{-1}{\left(\frac{-2M}{r} +\frac{r^2}{l^2} \right)}\left(a - \rd\frac{2Ml^2}{r^3} \right), \\
\rd &= -\left(d^2\left(\frac{2M}{r}-\frac{1}{l^2} \right) + a^2  \right)^\hf,\\
\xd &= \frac{b}{r^2},\\
\yd &=  \frac{c}{r^2},
\end{split}
\end{equation} 
where $d^2 = b^2 + c^2$ and we chose the negative root of $\rd^2$ as we want to look at a photon travelling tangentially at a distance $R$ from the origin falling into the black hole.
\begin{equation*}
\begin{split}
r(0) &= R, \\
\dot{r}(0) &= 0.
\end{split}
\end{equation*}
This allows us to directly solve for $a$ as
\begin{equation*}
a = - d\left(\frac{1}{l^2} - \frac{2M}{R} \right)^\hf. 
\end{equation*}
We have chosen the negative root as we would like a positive killing energy and $\td \ge 0$ so we can measure the co-ordinate time it takes to fall some distance towards the event horizon. Our equation thus becomes
\begin{equation}
\begin{split} \label{GDE}
\td &= \frac{drl^2}{\left(r^3-2Ml^2\right) }\left( \left( \frac{1}{l^2} - \frac{2M}{R^3}\right)^\hf - \frac{2Ml^2}{r^3} \sqrt{2M}\left( \frac{1}{r^3} - \frac{1}{R^3}\right)^\hf \right),\\
\dot{r} &=  -d\sqrt{2M}\left( \frac{1}{r^3} - \frac{1}{R^3}\right)^\hf, \\
\xd &=\frac{b}{r^2}, \\
\yd &=\frac{c}{r^2}. 
\end{split}
\end{equation}
So long as $R> r_+$ we have that these are all signed (away from $R$) quantities. This allows to deduce that $t,r,x,y$ are all monotonic. As $r$ is monotonic decreasing the geodesic never reaches $\mi$ and we find  that $\td,\xd,\yd$ also are monotonic. As $\gd$ is clearly smooth we deduce that $\gamma$ is a smooth embedding. \\

We now wish to study the co-ordinate time it takes to fall  a distance of $\frac{R}{2}$ remaining entirely outside the event horizon. More explicitly let $\frac{3}{2}r_+<\frac{R}{2}$ we are interested in
\begin{equation*}
\Delta_{\frac{R}{2}}t= t|_{r=R}- t|_{r=\frac{R}{2}},
\end{equation*}
To evaluate this we first need the affine time $\tau_{\frac{R}{2}}$ between $R$ and $\frac{R}{2}$. From \eqref{GDE} we can compute this as
\begin{equation*}
\tau_{\frac{R}{2}} =   \frac{1}{d\sqrt{2M}}\int_{\frac{R}{2}}^{R}\left( \frac{1}{r^3} - \frac{1}{R^3}\right)^{-\hf}dr,
\end{equation*} 
rescaling the integral we get
\begin{equation*}
\tau_{\frac{R}{2}}=\frac{R^{\frac{5}{2}}}{d\sqrt{2M}}\int_{{\hf}}^{1}\sqrt{\frac{y^3}{1-y^3}}dy.
\end{equation*}
Now 
\[
0< \int_{{\hf}}^{1}\sqrt{\frac{y^3}{1-y^3}}dy \le \int_{{0}}^{1}\sqrt{\frac{y^3}{1-y^3}}dy = \sqrt{\pi}\frac{\Gamma(\frac{5}{6})}{\Gamma(\frac{1}{3})} <\infty,
\]
so this integral is simply some positive constant. We denote it as $K_F$ and write
\begin{equation*}
\tau_{\frac{R}{2}} =\frac{K_FR^{\frac{5}{2}}}{d\sqrt{2M}}.
\end{equation*}
We now turn to the fall time. Our goal is to show we can for fixed $T$, find a geodesic that remains outside the event horizon up to and including $T$. In order to do this we consider bounding $\td$ below on the interval $r\in \left[\frac{R}{2}, R\right]$
\begin{equation*}
\begin{split}
\td &\ge   \frac{dl^2}{r^2} \left( \frac{1}{l^2} - \frac{2M}{R^3}\right)^\hf  -\frac{4dRl^2}{\left(R^3-16Ml^2\right) } \left( \frac{16Ml^2}{R^3} \sqrt{2M}\frac{\sqrt{7}}{R^{\frac{3}{2}}}\right)  \\
&\ge \frac{dl^2}{r^2}\left( {\frac{R^3-2Ml^2}{l^2R^3}}\right) ^\hf - \frac{4dl^2}{R^{\frac{7}{2}}(R^3-16Ml^2)},\\
\end{split}
\end{equation*}
now fix $\epsilon>1$ and set $C = l^{-1}\left(1-\frac{1}{\epsilon}\right)$ then provided $R^3\ge 2\epsilon Ml^2$ we have
\begin{equation*}
\begin{split}
&\td \ge \frac{dl^2C}{R^2} - \frac{4dl^2}{R^{\frac{7}{2}}(R^3-16Ml^2)}\\
&= \frac{dl^2}{R^2}\left(C - \frac{4}{R^\frac{3}{2} (R^3-16Ml^2)}\right). \
 \end{split}
\end{equation*} 
If we then choose $\epsilon=9$ so we require $R^3\ge 18Ml^2$ we have
\begin{equation*}
\begin{split}
\dot{r} \ge \frac{dl^2}{R^2}\left( \frac{8}{9l} - \frac{4}{R^\frac{3}{2}2Ml^2}\right), 
\end{split}
\end{equation*}
Insisting $R^3 \ge \max\{18Ml^2,\frac{100}{16M^2l^2}\}$ we find
\begin{equation*}
\td \ge \frac{4dl}{45} \frac{1}{R^2}.
\end{equation*}
We may now prove an inequality for the fall time (for large $R$)
\begin{equation*}
\Delta_{\frac{R}{2}}t = \int_0^{\tau_{\frac{R}{2}}}\td d\tau \ge K\frac{dl^2}{R^2}\cdot\tau_{\frac{R}{2}}= K_FR^\hf.
\end{equation*}
Now fix $T$ and let $R= \max\{\frac{9T^2}{4K_F^2},(18Ml^2)^\frac{1}{3},(\frac{100}{16M^2l^2})^\frac{1}{3} \}$ we then have
\begin{equation*}
\Delta_{\frac{R}{2}}t \ge \frac{3T}{2},
\end{equation*}
or in other words a geodesic with $r>\frac{3}{2}r_+$ for $0\le t\le T$.\\
\end{proof}
\section{Acknowledgements}
Jake Dunn is supported by EPSRC as part of the MASDOC DTC
at the University of Warwick. Grant No. EP/HO23364/1\\
The authors would also like to thank the anonymous referees for their comments. 
\bibliographystyle{ieeetr}
\bibliography{ref}
\end{document}